\documentclass[reqno,11pt]{amsart}

\IfFileExists{mymtpro2.sty}{%
  \usepackage[subscriptcorrection]{mymtpro2}
}{}

\usepackage{a4,amssymb}

\newtheorem{theorem}{Theorem}[section]
\newtheorem{lemma}{Lemma}[section]

\newtheorem{prop}{Proposition}[section]
\theoremstyle{definition}

\newcommand{\labelnummer}{\mbox{\normalfont (\roman{numcount})}}%

\makeatletter

  {\let\curlabelspeicher\@currentlabel%
    \begin{list}{\labelnummer}%
      {\usecounter{numcount}\leftmargin0pt%
        \topsep0.5ex\partopsep2ex\parsep0pt\itemsep0ex\@plus1\p@%
        \labelwidth2.5em\itemindent3.5em\labelsep1em%
      }%
    \let\saveitem\item%
    \def\item{\saveitem%
      \def\@currentlabel{{\upshape\curlabelspeicher}$\,$\labelnummer}}%
    \let\savelabel\label%
    \def\label##1{\savelabel{##1}%
      \@bsphack%
        \ifmmode\else%
          \protected@write\@auxout{}%
          {\string\newlabel{##1item}{{\labelnummer}{\thepage}}}%
        \fi%
      \@esphack%
    }%
  }{\end{list}}%

\renewcommand{\appendix}{\def\thesection{\textsc{Appendix}}}

 \let\leq\le
 \let\geq\ge

\let\Re\undefined \let\Im\undefined
\DeclareMathOperator{\Re}{Re}
\DeclareMathOperator{\Im}{Im}

\DeclareMathOperator{\tr}{tr\kern1pt}

\newcommand\EE{\mathbb E}

\makeatletter

\newif\ifper\pertrue
\def\per{.}

\def\bti{\@ifnextchar[\bbti\bbbti}
\def\bbti[#1]#2{#2, #1.}
\def\bbbti#1{#1.}

\def\z{\@ifnextchar[\zz\zzz}
\def\zz[#1]#2#3#4#5{\perfalse\emph{#2} \textbf{#3}, #4 (#5) [#1]}
\def\zzz#1#2#3#4{\emph{#1} \textbf{#2}, #3 (#4)\ifper\per\fi\pertrue}

\def\pub{\@ifstar\pubstar\pubnostar}
\def\pubnostar{\@ifnextchar[\@@pubnostar\@pubnostar}
\def\@@pubnostar[#1]#2#3#4{#2, #3, #4, #1\ifper\per\fi\pertrue}
\def\@pubnostar#1#2#3{#1, #2, #3\ifper\per\fi\pertrue}
\def\pubstar[#1]#2#3#4{\perfalse #2, #3, #4 [#1]\pertrue}

\makeatother

\newcommand{\bel}{\begin{equation} \label}
\newcommand{\ee}{\end{equation}}

\def\beq{\begin{equation}}
\def\eeq{\end{equation}}
\newcommand{\bea}{\begin{eqnarray}}
\newcommand{\eea}{\end{eqnarray}}
\newcommand{\beas}{\begin{eqnarray*}}
\newcommand{\eeas}{\end{eqnarray*}}

{

\newcommand{\Pp}{\mathbb{P}}
\newcommand{\R}{\mathbb{R}}

\newcommand{\Z}{\mathbb{Z}}

\newcommand{\N}{\mathbb{N}}

\newcommand{\C}{\mathbb{C}}
\newcommand{\E}{\mathbb{E}}
\newcommand{\mr}{}

\begin{document}

\title[Eigenvalue statistics for random point interactions]{Eigenvalue statistics for Schr\"odinger operators with random point interactions on $\R^d$, $d=1,2,3$}

\author[P.\ D.\ Hislop]{Peter D.\ Hislop}
\address{Department of Mathematics,
    University of Kentucky,
    Lexington, Kentucky  40506-0027, USA}
\email{peter.hislop@uky.edu}

\author[W.\ Kirsch]{Werner Kirsch}
\address{Fakult\"at f\"ur Mathematik und Informatik,
FernUniversit\"at in Hagen,
Universit\"atsstr.\ 1,
58084 Hagen, Germany}
\email{werner.kirsch@fernuni-hagen.de}

\author[M.\ Krishna]{M.\ Krishna}
\address{Ashoka University,
Haryana, India}
\email{krishna.maddaly@ashoka.edu.in}


\begin{abstract}
We prove that the local eigenvalue statistics at energy $E$ in the localization regime for Schr\"odinger operators with random point interactions on $\R^d$, for $d=1,2,3$, is a Poisson point process with the intensity measure given by the density of states at $E$ times the Lebesgue measure. This is one of the first examples of Poisson eigenvalue statistics for the localization regime of multi-dimensional random Schr\"odinger operators in the continuum. The special structure of resolvent of Schr\"odinger operators with point interactions facilitates the proof of the Minami estimate for these models.
\end{abstract}

\maketitle \thispagestyle{empty}

\begin{center}
{\it Dedicated to the memory of Alexandre Grossmann}
\end{center}

\tableofcontents

\vspace{.2in}



\section{Introduction: Random point interactions}\label{sec:intro1}
\setcounter{equation}{0}

Schr\"odinger operators with point interactions are useful models for many quantum phenomena. In this article, we continue our study of random Schr\"odinger operators with potentials formed from delta interactions at lattice points of $\Z^d$ and random coupling constants.
We study the formal Hamiltonian
\beq\label{eq:schr-op1}
H_\omega = - \Delta + \sum_{j \in \Z^d} \omega_j \delta (x - j),
\eeq
on $L^2 (\R^d)$, for $d = 1,2,3$. The coupling constants $\{ \omega_j ~|~ j \in \Z^d \}$ form a family of independent, identically distributed (iid) random variables with an absolutely continuous probability measure having a density $h_0 \in L_0^\infty (\R)$. Furthermore,
we assume the support of $h_0$ is the interval $[-b, -a]$ for some finite constants $0 < a < b < \infty$.
We refer to \cite{aghkh} for a detailed discussion of the construction of these operators via the Green's function formulae given in \eqref{eq:finite-green1} and \eqref{eq:kernel1}.

The family of random Schr\"odinger operators \eqref{eq:schr-op1} are covariant with respect to lattice translations so there exists a closed subset $\Sigma \subset \R$ so that the spectrum $\sigma (H_\omega) = \Sigma$ almost surely. Furthermore, under the conditions on the support of $h_0$, $\Sigma \cap \R^-$ is nonempty. We denote by $\Sigma_{pp}$ the almost sure pure point component of $\Sigma$.

In \cite{hkk1}, we proved that the family of random Schr\"odinger operators $H_\omega$ exhibits Anderson and dynamical localization at negative energies. We proved that there exists a finite energy $\tilde{E_0} < 0$ so that $\Sigma_{\rm pp} \cap (- \infty ,  \tilde{E_0} ]$ is almost surely nonempty.
Furthermore, for any $\phi \in L_0^2 (\R^d)$, any integer $q \in \N$, and any interval $I \subset (-\infty, \tilde{E_0}]$, with probability one we have
\beq\label{eq:dyn-loc1}
\sup_{t >0} \| \|x\|^{q/2} e^{-i t H_\omega} E_\omega (I) \phi \|_{\rm HS} < \infty ,
\eeq
where $\| A \|_{\rm HS}$ denotes the Hilbert-Schmidt norm of $A$. Bound \eqref{eq:dyn-loc1} is a hallmark of dynamical localization. We call the closure of the set of all such energies in $\Sigma_{\rm pp}$ for which \eqref{eq:dyn-loc1} holds the regime of complete localization denoted by $\Sigma^{\rm CL}$.

In this paper, we prove that the local eigenvalue statistics (LES) at any energy $E_0 < \tilde{E}_0$ in the localization regime is a Poisson point process with intensity measure $n(E_0)~ds$, where $n$ is the density of states (DOS). The existence of the DOS was proved in \cite[Corollary 7.3]{hkk1} as it follows from the Wegner estimate (see section \ref{subsec:wegner1}).

The main technical tool in \cite{hkk1} was the definition of finite-volume approximations $H_\omega^{L}$ to $H_\omega$ obtained by truncating \emph{only} the
potential to cubes $\Lambda_L$ of side width $L > 0$ and retaining the full Laplacian on $\R^d$.
These operators, acting on $L^2 ( \R^d)$, have the advantage that the resolvent of the Laplacian has an explicit kernel and that the resolvent is analytic on the negative real axis. These local operators are most easily described in terms of Green's functions. For $z \in \C \backslash [0, \infty)$, we let $G_0 (x,y;z)$ denote the Green's function of the Laplacian $H_0 := - \Delta$ on $\R^d$ corresponding to the resolvent $(H_0 - z)^{-1}$. Then, the Green's function $G_\omega^{L}(x,y;z)$, corresponding formally to the resolvent $(H_\omega^{L} - z)^{-1}$, is defined by
\beq\label{eq:finite-green1}
G_\omega^{L}(x,y;z) := G_0(x,y;z) + \sum_{i,j \in \tilde{\Lambda}_L} G_0(x,i;z) [K_{L} (z, \omega)^{-1}]_{ij} G_0(j,y;z) ,
\eeq
where $\tilde{\Lambda}_L := \Lambda_L \cap \Z^d$.
The kernel $K_{L}(z, \omega)$ is the random matrix on $\ell^2 ( \tilde{\Lambda}_L)$ given by
\beq\label{eq:kernel1}
[K_L (z, \omega)]_{ij} = \left( \frac{1}{\omega_{j,d}} - e_d(z) \right) \delta_{ij} - G_0(i,j; z)(1 - \delta_{ij} ) ,
\eeq
where the energy functions $e_d(z)$ and effective coupling constants $\omega_{j,d}$ are defined as in \cite[(4.7)-(4.9)]{hkk1}, and \cite[(4.10)-(4.11)]{hkk1}, respectively.

The family of operators $H_\omega^L$ constructed with the random potential restricted to the cube $\Lambda_L$ defined above suffices for the proof of localization at negative energies.
Local eigenvalue statistics (LES), however, requires more detailed information on the negative eigenvalues of local Hamiltonians. In section \ref{sec:bc1}, we describe local Hamiltonians $H_\omega^{\Lambda_L}$ on cubes in $\Lambda_L \subset \R^d$ obtained by imposing Dirichlet boundary conditions on the boundary $\partial \Lambda_L$. This theory was developed by Blanchard, Figari, and Mantile \cite{blanchard}, and by Pankrashkin \cite{pankrashkin}.
The LES associated with $H_\omega^{\Lambda_L}$ are defined as follows. Let $E_j^{\Lambda_L}(\omega)$ denote the eigenvalues of $H_\omega^{\Lambda_L}$. For $E_0 \in (-\infty, \tilde{E_0}] \cap \Sigma^{\rm CL}$, we define the point process $\xi_\omega^L$ on $\R$ by \beq\label{eq:les1}
\xi_\omega^L(ds) := \sum_{j} \delta (|\Lambda_L|( E_j^L(\omega) - E_0) - s) ~ds .
\eeq
These point processes and their limit points for $L \rightarrow \infty$ were studied by Minami \cite{min1} for lattice models on $\Z^d$ and at energies $E_0$ in the localization regime. Minami's work was partially inspired by earlier work of Molchanov \cite{molchanov} who proved Poisson eigenvalue statistics for a one-dimensional model on $\R$. The results for lattice models were later elaborated upon and extended by Germinet and Klopp \cite{germinet-klopp}.

Localization plays a key role in these proofs. In order to describe this,
we need to introduce a second family of point processes based on a second scale $\ell := L^\alpha$, for $0 < \alpha < 1$ so that $\ell \ll L$.
We divide $\Lambda_L$ into subcubes $\Lambda_\ell^p$, and consider the family of local random Schr\"odinger operators $H_\omega^{\ell,p}$, obtained by restrictions to the smaller cubes $\Lambda_\ell^p$, with eigenvalues $E_j^{\ell,p}(\omega)$. For each $p$, we define a local point process $\eta^{\ell,p}_\omega$ associated with the eigenvalues $E_j^{\ell,p}(\omega)$ as in \eqref{eq:les1}, that is
\beq\label{eq:les9}
\eta_\omega^{\ell,p}(ds) := \sum_{j} \delta (|\Lambda_L|( E_j^{\ell,p}(\omega) - E_0) - s) ~ds\, .\eeq
 For $p \neq q$, the two processes $\eta^{\ell,p}_\omega$ and $\eta^{\ell,q}_\omega$ are independent.
Localization estimates are used to prove that the point process $\xi_\omega^L$, constructed from the eigenvalues of $H_\omega^L$ as in \eqref{eq:les1}, and the point process $\zeta_\omega^L$, constructed from the array $\eta_\omega^{\ell, p}$ of independent process, each formed from the eigenvalues of the operators $H_\omega^{\ell,p}$, have the same $L \rightarrow \infty$ limit. The analysis of the limiting point process associated with the array $\eta_\omega^{\ell, p}$ depends on a new Minami estimate. A Minami estimate is possible for this model on $\R^d$ because of the special structure of the point interactions. 

In order to formulate the main result of this paper, we note that the family of random Schr\"odinger operators $H_\omega$ with point interactions,  formally described in \eqref{eq:schr-op1} on $\R^d$, for $d=1,2,3$, is defined as the self-adjoint operator associated with the resolvent expression \eqref{eq:finite-green1} and \eqref{eq:kernel1} with $\tilde{\Lambda}_L$ replaced by $\Z^d$. This corresponds to the choice of a self-adjoint extension of $H_0 = - \Delta$ on $C_0^\infty ( \R^d \backslash \Z^d)$. We refer to the book \cite{aghkh} for the details of this construction. The local operators $H^L_\omega$ and $H^{\ell,p}_\omega$ are constructed in section \ref{sec:bc1} following \cite{blanchard, pankrashkin}. 
We also need a regime of localization of the almost sure spectrum $\Sigma$. The existence of this region at negative energies was proved in \cite{hkk1} under the following assumption on the random variables:

\begin{description}
\item [{[H1]}] 
The coupling constants $\{ \omega_j ~|~ j \in \Z^d \}$ form a family of independent, identically distributed (iid) random variables with an absolutely continuous probability measure having a density $h_0 \in L_0^\infty (\R)$. Furthermore,
we assume the support of $h_0$ is the interval $[-b, -a]$ for some finite constants $0 < a < b < \infty$.
\end{description}

\begin{theorem}\label{thm:main1}
Let $H_\omega$ be the family of random Schr\"odinger operators with point interactions formally defined in \eqref{eq:schr-op1} on $\R^d$, for $d=1,2,3$. In addition, suppose that the random coupling constants $\{ \omega_j \}_{\Z^d}$ satisfy hypothesis [H1]. 
Then, for any fixed energy $E_0 \in \Sigma_{\rm pp} \cap (-\infty, \tilde{E_0}] \subset \Sigma^{\rm CL}$ for which the density of states is  positive: $n(E_0) > 0$, the local eigenvalue statistics $\xi_\omega^L$ in \eqref{eq:les1} converges weakly to a Poisson point process with intensity measure $n(E_0) ds$.
\end{theorem}

This is one of the first results on the nature of the limiting eigenvalue point process of the LES for random Schr\"odinger operators on $\R^d$, for $d \geq 2$. In \cite{hk1}, two of the authors proved that the limit points of the local processes $\xi_\omega^L$ are compound Poisson point processes. In the absence of a Minami-type estimate, this is the best possible result. The special nature of the point interactions makes a Minami-type estimate possible for the models studied in this paper. As a result, we are able to establish that the LES are Poisson.
During the completion of this paper, Dietlein and Elgart \cite{DietleinElgart2017} proved a Minami estimate for random Schr\"odinger operators on $\R^d$ with Anderson-type random potentials at energies near the bottom of the deterministic spectrum. They used this estimate to prove Poisson statistics in that energy regime. It is not clear if their results can be adapted in order to treat the random delta potential models considered here.


\subsection{Contents of the paper}\label{subsec:contents}

In section \ref{sec:bc1}, we review the formulation of Schr\"odinger operators with point interactions $H^L_\omega$ on bounded domains in terms of Green's functions. Finite-volume estimates for the random Schr\"odinger operators on bounded domains are proved in section \ref{sec:w-m-est1}, including the Minami estimate. A new proof of the effect of rank one perturbations is presented in section \ref{sec:trace-est1} that may be of independent interest. The properties of certain arrays of independent random variables are presented in section \ref{sec:les-indep-arrays1}.
A main technical result of this paper is that the process $\zeta_\omega^L$, constructed from the array $\eta_\omega^{\ell, p}$ of independent process, has the same limit as the LES $\xi_\omega^L$. This is proved in section \ref{sec:approx-pt-proc1}. The analysis of $\xi_\omega^L$ and a proof of Theorem \ref{thm:main1} is presented in section \ref{sec:sp-stat1}. All of the explicit claculations in sections \ref{sec:les-indep-arrays1} and \ref{sec:approx-pt-proc1} are presented for the case $d=3$. Comments on the case $d=1$ and $d=2$ are presented in section \ref{sec:dimensions12}.


\section{Local operators for point interactions with boundary conditions}\label{sec:bc1}
\setcounter{equation}{0}

Although the calculations with $H_\omega = - \Delta + V_\omega^L$ on $L^2 (\R^d)$
suffice to prove localization at negative energies, local eigenvalue statistics
(LES) requires finer estimates. In particular, LES require that we work with operators localized to boxes $\Lambda_L$ and acting on the Hilbert space $L^2 (\Lambda_L)$. The theory of point interactions in a bound domain $\Omega \subset \R^d$, for $d=1,2,3$ is developed in \cite{blanchard, pankrashkin}. We begin with a general formulation on a bounded domain $\Omega \subset \R^d$, with a piecewise smooth boundary, and with $N$ delta interactions at distinct points $\{ x_k \}_{k=1}^N \subset \Omega$. In the text, \mr we will explicitly treat the case $d =3$. The necessary formulae for $d=1$ and $d=2$ are given in section \ref{sec:dimensions12}.

The Green's function for $\Omega \subset \R^3$ with $X$-boundary conditions is
\beq\label{eq:green1}
G_\Omega^X(x,y;z) = G_0(x,y;z) -c_{z,y}^X(x), ~~x,y \in \Omega ,
\eeq
where $G_0$ is the free Green's function on $\R^3$ given by
\beq\label{eq:free-green1}
G_0(x,y;z) = \frac{e^{-i \sqrt{z} \| x-y \|}}{4 \pi \|x-y\|}.
\eeq
For fixed $y \in \Omega$, the corrector $c_{z,y}^{\Omega,X}(x)$ satisfies the boundary-value problem:
\beq\label{eq:corrector-bvp1}
((- \Delta_x - z ) c^{\Omega,X}_{z,y})(x) = 0 , ~~x \in \Omega
\eeq
with Dirichlet or Neumann boundary conditions on $\partial \Omega$:
\bea\label{eq:corrector-bvp2}
c_{z,y}^{\Omega,D}(x) |_{\partial \Omega} & = &  G_0(x,y;z) |_{x \in \partial \Omega} ; \\
\nu \cdot \nabla_x c_{z,y}^{\Omega,N}(x) |_{\partial \Omega} & = &   \nu \cdot \nabla_x G_0(x,y;z) |_{x \in \partial \Omega}.
\eea
When the region $\Omega$ is a cube $\Lambda_L = [0,L]^d$, we can also consider periodic boundary conditions. The periodic Green's function is obtained by periodizing $G_0$:
\beq\label{eq:periodic-gr1}
G_{\Lambda_L}^P (x,y ;z) = \sum_{m \in \Z^d}  G_0 (x + mL, y,;z) =  \sum_{m \in \Z^d}  g_0 (\|x - y + mL\|;z) ,
\eeq
where $g_0(s) = (4 \pi s)^{-1}{e^{-i \sqrt{z} s}}.$
With $y \in \Lambda_L$ fixed, we may write $G_{\Lambda_L}^P$ as in \eqref{eq:green1} with
\beq\label{eq:corrector-bvp3}
c_{z,y}^{\Lambda_L,P}(x) |_{\partial \Lambda_L}  =  [ G^P_{\Lambda_L}(x,y;z) - G_0(x,y;z)] |_{x \in \partial \Lambda_L} .
\eeq

We now consider $N$ delta functions located at points $\{x_k\}_{k=1}^N \subset \Omega$ with real coupling constants $\alpha := \{\alpha_k\}_{k=1}^N$. 
The  Hamiltonian
$H_\alpha^{\Omega,X}$, for Dirichlet, Neumann, or periodic boundary conditions indicated by $X=D,N,P$,  on a bounded region $\Omega$ corresponding to this delta function potential is formally defined as
\beq\label{eq:local-delta-interaction1}
H_\alpha^{\Omega,X} = - \Delta_\Omega^X + \sum_{k=1}^N \alpha_k \delta(x - x_k) ,
\eeq
where $H_0^{\Omega,X} = - \Delta_\Omega^X $ is the Laplacian on $\Omega$ with $X$ boundary conditions.
The operator \eqref{eq:local-delta-interaction1} is made rigorous through the choice of the appropriate self-adjoint extension. As a result of this analysis, the Green's function $G_{\alpha}^{\Omega,X}(x,y;z)$ for $H_\alpha^{\Omega,X}$ and $d=3$ is related to the Green's function $G_\Omega^X(x,y;z)$ for the unperturbed operator $H_0^{\Omega,X}$ by
\beq\label{eq:green2}
G_{\alpha}^{\Omega,X}(x,y;z) = G_0^{\Omega, X}(x,y;z) + \sum_{j,k=1}^N G_0^{\Omega, X}(x,x_j;z) [ K_\Omega^X (z;\alpha)^{-1} ]_{jk} G_0^{\Omega, X}(x_k,y;z).
\eeq

The operator $K_\Omega^X (z;\alpha): \C^N \rightarrow \C^N$ is an $N \times N$ matrix-valued function defined by
\beq\label{eq:K-kernel1}
[K_\Omega^X (z;\alpha)]_{jk} := - G_0^{\Omega, X}(x_j,x_k;z)(1 - \delta_{jk}) + \left( \frac{1}{\alpha_{d,k}} + c_{z, x_k}^{\Omega, X}(x_j) - e_d(z) \right) \delta_{jk}.
\eeq
The effective energy $e_d(z)$ is dimension dependent and defined as in \cite[(4.7)-(4.9)]{hkk1}. For $d=3$, we have:
\beq\label{eq:effective-energy1}
e_3(z) = \frac{i \sqrt{z}}{4 \pi},
\eeq
where the square root function is defined with the branch cut along the positive real axis.
The effective coupling constants $\alpha_{d,k}$ are also dimension-dependent
and $\alpha_{3,k} =  \alpha_k$.

It is convenient to think of $K_\Omega^X (z; \alpha)^{-1}$ as the resolvent of a generalized matrix Schr\"odinger operator $h_\alpha^{\Omega,X} (z)
:= t^{\Omega, X}(z) + v$ at energy $e_{d}(z)$ on $\C^N$. The kinetic energy $t^{\Omega, X}(z)$ has matrix elements
\beq\label{eq:matrixKE1}
t^{\Omega, X}_{jk}(z) :=    c_{z, x_k}^{\Omega, X}(x_j) \delta_{jk} - G_0^{\Omega, X}(x_j,x_k;z)(1 - \delta_{jk}) ,
\eeq
and the potential $v$, depending on the coupling constants $\alpha_{d,k}$, is a diagonal matrix given by
\beq\label{eq:matrixKE2}
v_{jk} := \frac{1}{\alpha_{d,k}} \delta_{jk}.
\eeq
The off-diagonal part of $t^X(z)$ decays exponentially for $z \not\in \sigma (H_\alpha^X)$. The diagonal part of $t^X(z)$ is uniformly bounded provided the constants $\alpha_{d,k}$ are uniformly bounded away from zero.


\section{Finite-volume estimates: Wegner, Minami, and localization estimates}\label{sec:w-m-est1}
\setcounter{equation}{0}

We apply \mr results of the last section on point interactions in bounded domains to local random Schr\"odinger operators with point interactions.
Let $\Lambda_L \subset \R^d$, for $d=1,2,3$, be a cube of side length $L > 0$. We denote by $H_\omega^L$ the restriction of the operator defined in \eqref{eq:schr-op1}
to $\Lambda_L$ with $X$ boundary conditions on $\partial \Lambda_L$.
In this setting, the number of point interactions satisfies $N \sim L^d$ and the coupling constants $\alpha_k = \omega_k$ are random variables.

For the remaining sections, we work explicitly with Dirichlet boundary conditions $X = D$ and dimension $d=3$.
The regions $\Omega$ will be cubes $\Lambda_L$ of side length $L>0$.  To simplify notation, we write $H_0^L = H_0^{\Lambda_L}=H_0^{\Lambda_L,D}$ and $H_\omega^L = H_\omega^{\Lambda_L}=H_\omega^{\Lambda_L,D}$. We also denote by $G_0^L (x,y;z)$ the free Green's function for $H_0^L$, and by $G_\omega^L(x,y;z)$ the interacting Green's function for $H_\omega^L$.

  In order to prove that the local eigenvalue statistics converge to a Poisson process, we need a Wegner and Minami estimate for the local Hamiltonian $H_\omega^{L}$. Although these operators act on $L^2 (\Lambda_L)$, the special structure of the potential allows us to compute a Minami estimate following the ideas of the proof in \cite{cgk1}. In particular, spectral averaging in the trace norm, rather than simply for matrix elements, is possible for this model.

\subsection{Wegner estimate}\label{subsec:wegner1}

We first review the Wegner estimate proved in \cite[Theorem 7.1]{hkk1}. Although this theorem was proved in \cite{hkk1}, we present a different proof here. The reason for this is that the techniques will be used in the proof of the Minami estimate that was not proved in \cite{hkk1}.
We write $E_{H_\omega^{\Lambda}} (I)$ for the spectral projector of $H_\omega^\Lambda$ and an interval $I \subset \R$.

\begin{theorem}\label{thm:wegner1}
For any $E_0 < 0$ there is a constant $C_{W}<\infty$ such that for any interval $I:= [E-\eta,E+\eta] \subset ( - \infty, E_{0})$ (with $\eta>0$)
\bea\label{eq:wegner1}
\Pp \{ {\rm dist} ( \sigma (H_\omega^{L}),  E) < \eta \} &=& \Pp \{ {\rm Tr} E_{H_\omega^{\Lambda_L}} (I) \geq 1 \} \nonumber \\
  & \leq & \E \{ {\rm Tr} E_{H_\omega^{\Lambda_L} }(I) \} \nonumber \\
  &\leq & C_W |\Lambda_L| \eta .
\eea
\end{theorem}

The basic tool in the proof of Theorem \ref{thm:wegner1} is the spectral averaging result for one-parameter variations.
Even though a spectral averaging in the trace class is not possible for the usual random Schr\"odinger operator on $\R^d$, the special structure of the point interaction model makes this possible. We define the random variable $X_{\omega_j, \omega_j^\perp}^{L}(I) := {\rm Tr} E_{H_\omega^{\Lambda}} (I)$, where we write $L$ for the cube $\Lambda_L$.

\begin{lemma}\label{lemma:sp-ave1}
We consider the variation of  the Hamiltonian $H_\omega^\Lambda$ with respect to one random variable $\omega_j$, for $j \in \tilde{\Lambda}$, with all the other random variables $\omega^\perp_j$ held fixed, so that $\omega := (\omega_j, \omega_j^\perp)$. For any $I \subset \R^-$, we then have
\beq\label{eq:sp-ave-tr1}
 \E_{\omega_j} \{ X_{\omega_j, \omega_j^\perp}^{L}(I) \} \leq C_W |I| |\Lambda|,
\eeq
for a constant $C_W(I) > 0$ depending only on $\sup I$.
\end{lemma}

\begin{proof}
\noindent
1. We use Stone's formula to express the spectral projection $E_{H_\omega^{L}} (I)$ as an integral over the resolvent:
\beq\label{eq:contour1}
E_{H_\omega^{L}} (I) = \frac{1}{\pi} \lim_{\epsilon \rightarrow 0}
\int_{I} ~ \Im R_\omega^{L}(E+i \epsilon) ~dE .
\eeq
Because $R_0(z)$ is analytic away from $\R^+$, the resolvent formula \eqref{eq:finite-green1} yields
\beq\label{eq:contour2}
E_{H_\omega^{L}} (I) =  \frac{1}{\pi} \lim_{\epsilon \rightarrow 0}
\sum_{\ell,m \in \tilde{\Lambda}_L} \int_I ~ \Im [ R_0^L(\cdot, \ell;E) [K_{L} (E+i \epsilon; \omega_j)^{-1} ]_{\ell m} R_0^L(m, \cdot;E) ]  .
\eeq
The trace is expressible as the integral over the diagonal of the corresponding Green's functions (as may be justified using the Hilbert-Schmidt bound on the trace):
\beq\label{eq:contour3}
X_{\omega_j, \omega_j^\perp}^{L}(I) =
\frac{1}{\pi} \int_{I} \lim_{\epsilon \rightarrow 0}
\sum_{\ell,m \in \tilde{\Lambda}_L} \int_{\Lambda_L} ~ G_0^L(x,\ell;E) \Im [K_{L} (E+i \epsilon; \omega_j)^{-1} ]_{\ell m} G_0^L(m,x;E) ~dE ~d^3x,
\eeq
since the Green's functions are real for $z \in ( - \infty, 0)$.

\noindent
2. Only the kernel $K$ in \eqref{eq:contour3} depends on $\omega_j$. In \cite[section 5]{hkk1}, we developed a spectral averaging method that is applicable to $K_L^{-1}(z, \omega)$, using the differential inequality method of \cite{chm1}. We proved that for any $\xi \in \ell^2 ( \tilde{\Lambda}_L)$, there is a constant $C _1 > 0$ so that
\beq\label{eq:sp-ave1}
\sup_{\epsilon \rightarrow 0} \left| \int h_0 (\omega_j) \langle \xi ,  [K_{L} (E + i \epsilon; (\omega_j, \omega_j^\perp))]^{-1} \xi \rangle ~d \omega_j \right|
\leq C_1 \| \xi \|^2 .
\eeq
 Upon taking the expectation of \eqref{eq:contour3} with respect to $\omega_j$, and using the spectral averaging
result \eqref{eq:sp-ave1} with the vector $\xi$ having components $\xi_m = G_0(m,x;E)$, we obtain\mr
\beq\label{eq:contour4}
\E_{\omega_j} \{  X_{\omega_j, \omega_j^\perp}^{(L)}(I) \} \leq C_1 \sum_{m \in \tilde{\Lambda}_L}
\frac{1}{ \pi} \int_{I} \int_{\Lambda_L} ~ G_0^L(x,m;E)^2  ~d^3x~dE.
\eeq
We use the representation of the local Green's function given in \eqref{eq:green1}. The fact that $E \not \in \rho (H_0^L)$ implies that the Green's function is exponentially decaying yielding a bound for the $x$-integral that is $\mathcal{O}(1)$. The sum over $m \in \tilde{\Lambda}_L$ gives a factor of $|\Lambda_L|$.
So, after the trivial $E$-integration, we obtain
\beq\label{eq:sp-ave2}
\E_{\omega_j} \{  X_{\omega_j, \omega_j^\perp}^{L}(I) \} \leq C_1 |\Lambda_L| |I|,
\eeq
where the constant $C_1$ is uniform with respect to $\omega_j^\perp$.
\end{proof}

This lemma immediately implies the Wegner estimate.

\begin{proof}[Proof of Theorem \ref{thm:wegner1}]
Given the spectral averaging result, Lemma \ref{lemma:sp-ave1}, the proof of the Wegner estimate follows:
\bea\label{eq:wegner2}
\Pp \{ {\rm dist} ( \sigma (H_\omega^{L}),  E_0 ) < \eta \} &=& \Pp \{ {\rm Tr} E_{H_\omega^{\Lambda_L}} (I_\eta) \geq 1 \} \nonumber \\
  & \leq & \E \{ {\rm Tr} E_{H_\omega^{\Lambda_L}} (I_\eta) \} \nonumber \\
  &\leq & \E_{\omega_j^\perp} [ \E_{\omega_j} \{  X_{\omega_j, \omega_j^\perp}^{L}(I) \} ] \nonumber \\
  & \leq & C_W |\Lambda_L| \eta ,
\eea
by \eqref{eq:sp-ave2} since $| I_\eta| = 2 \eta$.
\end{proof}

\subsection{Minami estimate}\label{subsec:minami1}

As in section \ref{subsec:wegner1}, we define the random variable $X_\omega^{L}(I)$ by
\beq\label{eq:trace1}
X_\omega^{L}(I) := {\rm Tr} E_{H_\omega^{L}} (I),
\eeq
for an interval $I \subset  ( -\infty, 0)$. The trace is well-defined since the negative spectrum consists of finitely-many eigenvalues with finite multiplicities.
The Minami estimate consists in estimating $\E \{  X_\omega^{(L)}(I) ( X_\omega^{(L)}(I) - 1) \}$.
Until recently, there was no known Minami estimate for random Schr\"odinger operators on $L^2(\R^d)$, for $d \geq 2$ (see \cite{DietleinElgart2017}, the methods developed there are not immediately applicable to the models discussed here.)
The special structure of delta interactions allows modification of the proof in \cite{cgk2} for lattice models.

\begin{theorem}\label{thm:minami-est1}\mr
For any $E_0 < 0$ there is a constant $C_{M}<\infty$ such that for any interval $I:= [E-\eta,E+\eta] \subset ( - \infty, E_{0})$

\beq\label{eq:minami1}
 \E \{  X_\omega^{L}(I) ( X_\omega^{L}(I) - 1) \} \leq C_M  |\Lambda_L|^2 \eta^2 .
\eeq
\end{theorem}

\begin{proof}
1. {\it One-parameter perturbation}. The key to dealing with point interactions is that the variation of one parameter, say $\omega_j$, results in a rank one perturbation.
To see this, it follows from \eqref{eq:green2} that the difference of the resolvents of the matrix Schr\"odinger operators $h_\omega^L(z)$ for a variation of $\omega_j$ to $\tau_j$ is given by
\bea\label{eq:eq:rank-oneK1}
\lefteqn{ [K_L(z, (\omega_j, \omega_j^\perp))^{-1} - [K_L(z,(\tau_j, \omega_j^\perp))^{-1}]_{km} } \nonumber \\
 & = &\sum_{p,q\in \tilde{\Lambda}_L}\; [K_L((\omega_j, \omega_j^\perp),z)^{-1}]_{kp} \{ K_L(z,(\tau_j, \omega_j^\perp))_{pq} - K_L(z,(\omega_j, \omega_j^\perp))_{pq} \}
 [ K_L(z,(\tau_j, \omega_j^\perp))^{-1}]_{qm} \nonumber \\
 & = &\sum_{p,q\in \tilde{\Lambda}_L} \;[K_L(z,(\omega_j, \omega_j^\perp))^{-1}]_{kp} \left\{ \left( \frac{\omega_j - \tau_j}{\omega_j \tau_j} \right)
 \delta_{pj} \delta_{qj} \right\}  [K_L(z,(\tau_j, \omega_j^\perp))^{-1}]_{qm}
 \nonumber \\
 &=& \left( \frac{\omega_j - \tau_j}{\omega_j \tau_j} \right) [K_L(z,(\omega_j, \omega_j^\perp))^{-1}]_{kj}  [K_L(z,(\tau_j, \omega_j^\perp))^{-1}]_{jm}
\eea
This shows that the difference $K_{L} (z; \omega_j)^{-1}- K_{(L)} (z; \tau_j)^{-1} $
is a rank-one matrix.
Substituting this into the relation for the difference of the resolvents,
and writing only $\omega_j$ for $\omega$ with $\omega_j^\perp$ fixed, we find:
\bea\label{eq:rank-one2}
\lefteqn{ R_{\omega_j}^{L}(z) - R_{\tau_j}^{L}(z) } \nonumber \\
 &=& \sum_{k,m \in \tilde{\Lambda}_L} R_0^L (\cdot, k;z) [K_{L} (z; \omega_j)^{-1} - K_{L} (z; \tau_j)^{-1}]_{km} R_0^L (m,\cdot;z) \nonumber \\
 & = & C(\omega_j, \tau_j) \left( \sum_{k \in \tilde{\Lambda}_L} R_0^L(\cdot,k;z) K_L(z,(\omega_j, \omega_j^\perp))^{-1}_{jk} \right) \left(  \sum_{m \in \tilde{\Lambda}_L} [K_L(z,(\tau_j, \omega_j^\perp))^{-1}]_{jm} R_0^L(\cdot, m;z) \right), \nonumber \\
  & &
\eea
where the constant is
$$
C(\omega_j, \tau_j) := \left( \frac{\omega_j - \tau_j}{\omega_j \tau_j} \right).
$$
and the difference $K^{(L)} (z; \omega_j)^{-1}- K^{(L)} (z; \tau_j)^{-1} $
is a rank-one operator.
This shows that the resulting change in the resolvents is a rank one perturbation.

\noindent
2. {\it Estimate on the eigenvalue counting function.}
Since the Hamiltonians $H_\omega^L$ are lower semibounded, we can choose a real energy $E \ll 0$ so that
$z =  E \ll \inf \Sigma$.  With this choice of $z$, the resolvent $R_\omega^L(z)$ is a self-adjoint operator.
Let $I = (a,b)$. It follows that $H_\omega^L$ has an eigenvalue in $I$ if and only if $R_\omega^L(z)$
has an eigenvalue in $I_z : = ( (b-z)^{-1}, (a-z)^{-1})$. Consequently, the eigenvalue counting functions satisfy
\beq\label{eq:evcounting1}
X_\omega^L(I) := {\rm Tr} E_{H_\omega^L}(I) =  {\rm Tr} E_{R_\omega^L(z)}(I_z) .
\eeq
We now consider two configurations $(\omega_j, \omega_j^\perp)$ and $(\tau_j, \omega_j^\perp)$ obtained by varying the random variable in the $j^{th}$-position. We find
\beq\label{eq:evcounting2}
X_{\omega_j}^L(I) - X_{\tau_j}^L(I) = {\rm Tr} E_{R_{\omega_j}^L(z)}(I_z) -  {\rm Tr} E_{R_{\tau_j}^L(z)}(I_z).
\eeq
As shown in \eqref{eq:rank-one2}, the difference of the resolvents $R_{\omega_j}^L(z) - R_{\tau_j}^L(z)$
is a rank one operator. From Proposition \ref{prop:rank-one1}, it follows that
\bea\label{eq:rank-one1}
| X_{\omega_j}^{L}(I) -  X_{\tau_j}^{L}(I)| & = & | {\rm Tr} E_{R_{\omega_j}^L(z)}(I_z) -  {\rm Tr} E_{R_{\tau_j}^L(z)}(I_z) | \nonumber \\
 & \leq & 1 .
\eea
If, for example, $X_{\omega_j}^{L}(I) \geq 1$, then
\beq\label{eq:rank-one3}
 X_{\omega_j}^{L}(I)  \leq  X_{\tau_j}^{L}(I) + 1.
 \eeq


\noindent
3. \emph{Conclusion of the proof.} Following the ideas of \cite{cgk1}, we take $0 < c < d$ so that the interval $[-d, -c]$ is disjoint from $[-b, -a]$. We define a new random variable $\tau_j \in [-d, -c]$, having the same distribution as $\omega_j$. From \eqref{eq:rank-one1} and \eqref{eq:sp-ave2}, we obtain,
\bea\label{eq:minami2}
\E \{  X_\omega^{L}(I) ( X_\omega^{L}(I) - 1) \} & \leq & \E_{\tau_j} \E \{ X_{\omega_j, \omega_j^\perp}^{L}(I)
( X_{\tau_j, \omega_j^\perp}^{L}(I)) \} \nonumber \\
 & \leq & C_1 |\Lambda_L| |I| \left( \E_{\tau_j} \E_{\omega_j^\perp} \{ X_{\tau_j,\omega_j^\perp}^{L}(I) \}
\right)
\eea
%
\mr
Since $\tau_j$ is a random variable with the same probability density $h_0$ and independent of $\omega_j^\perp$,
we apply the usual Wegner estimate to evaluate the expectation in the right hand side of \eqref{eq:minami2} with respect to the variables $(\tau_j, \omega_j^\perp)$. As a consequence, we have
\beq\label{eq:minami4}
\E_{\tau_j} \E_{\omega_j^\perp} \{ X_{\tau_j, \omega_j^\perp}^{L}(I) \} \leq C_W |\Lambda_L| |I|.
\eeq
This estimate, together with \eqref{eq:minami2}, proves the theorem.
\end{proof}


\subsection{Localization estimates}\label{subsec:loc1}

In \cite[section 4]{hkk1}, we proved exponential decay of the fractional moments of the Green's function at negative energies, a key step in the proof of localization. For LES, we work with the finite-volume operators and need the decay of the matrix operator $K_{\Lambda} (z, \omega)^{-1}$. We briefly review those results here.
The random matrix operator $K_\Lambda(z, \omega)$ on $\ell^2 (\Lambda)$ is a generalized random Schr\"odinger operator $h_\omega^\Lambda(z)$ as defined in \eqref{eq:matrixKE1} and \eqref{eq:matrixKE2}.
The kinetic energy operator $t(z)$ on $\ell^2 ( \Lambda)$ is defined by
\beq\label{eq:ke1}
t^\Lambda_{jk}(z) :=    c_{z, x_k}^\Lambda(x_j) \delta_{jk} - G_0^\Lambda(x_j,x_k;z)(1 - \delta_{jk}) ,
\eeq
A diagonal, local random potential $v_\omega^\Lambda$ is defined by
\beq\label{eq:pe1}
[v_\omega^\Lambda]_{jk} := \frac{1}{\alpha_{d,k}} \delta_{jk}.
\eeq
The discrete generalized Schr\"odinger operator $h_\omega^\Lambda(z)$ is defined by
$$
h_\omega^\Lambda (z) := t^\Lambda(z) + v_\omega^\Lambda .
$$
Then, the random matrix $K_\Lambda (z,\omega)$ is given by
$$
K_\Lambda(z, \omega) = h_\omega^\Lambda (z) - e_d(z), ~~ z \not \in [0, + \infty).
$$

\begin{prop}\cite[Proposition 4.7]{hkk1}\label{prop:loc1}
 There exists an energy $\tilde{E}_0 < 0$, with $\Sigma_{pp} \cap ( - \infty, \tilde{E}_0] \neq \emptyset$ so that for any $s \in (0,1)$ and $z \in \C$, with $\Re z < \tilde{E}_0 < 0$ and $\frac{\pi}{2} < \arg z < \frac{3\pi}{2}$, there are finite, positive constants $C_s(z) > 0$ and $\gamma_{s,d} (z) >0$, uniform in $L > 0$, and locally uniform in $z$,
so that we have
\beq\label{eq:K-exp-decay1}
\E \{ | [K_\Lambda (z ,\omega)^{-1}]_{ij} |^s \} \leq C_s (z) e^{-  \gamma_{s,d} (z) \|i-j\| } ,
\eeq
for any $i, j \in \tilde{\Lambda}_L$. For $z = E + \zeta$, with $E < \tilde{E}_0 < 0$ and $|\zeta|$ small, the exponent $\gamma_{s,d} (z) \approx |E|^{\frac{1}{2}}$.   
\end{prop}


\section{Local eigenvalue statistics: independent arrays}\label{sec:les-indep-arrays1}
\setcounter{equation}{0}

Minami \cite{min1} realized that localization (in the fractional moment sense) implied that the process $\xi_\omega^L$ may be approximated by a point process constructed from the eigenvalues of Hamiltonians on a smaller scale $\ell$. Since these Hamiltonians are local, the sets of their eigenvalues are independent so that  the corresponding point processes are independent. Consequently, one can apply the well-developed theory of limiting processes of independent arrays of point processes, see, for example \cite[chapter 11]{daley-vere-jones1}.

We define these arrays as follows. For integers $L , \ell$ so that $\ell$ divides $L$, we divide the cube $\Lambda_L$ into subcubes $\Lambda_\ell^p$ of side length $\ell$ centered at points $p \in \Z^d$.
There are $N_L = (L / \ell)^d$ such subcubes and $\Lambda_L = \cup_p \Lambda_\ell^p$, up to a set of measure zero.
We denote by $H_\omega^{\ell,p}$ the local point interaction Hamiltonian restricted to $\Lambda_\ell^p$ with Dirichlet boundary conditions. We compare the eigenvalue statistics associated with $H_\omega^L$ and $H_\omega^\ell := \oplus_p H_\omega^{\ell, p}$ at negative energies.
The resolvent of $H_\omega^L$ is denoted by $R_L(z)$ and of $H_\omega^{\ell,p}$ by $R_\omega^{\ell,p}(z)$, so that the resolvent of $R_\omega^{\ell}(z)$ is written $R^\ell_\omega(z) = \oplus_p R_\omega^{\ell,p}(z)$.

We always choose $L$ so that the boundaries of the cubes $\Lambda_L$, denoted $\partial \Lambda_L$, do not intersect the lattice $\Z^d$ and the distance from $\partial \Lambda_L$ to $\Z^d$ is, say, $\frac{1}{2}$.

The local operators $H_\omega^L$ and $H_\omega^{\ell,p}$ have discrete spectrum.
In addition, the spectra of the local Laplacians are semi-lower bounded and lie in the half-axis $[\Sigma_0, \infty )$, for some $- \infty < \Sigma_0 < 0$, independent of $L$.
We recall that the Wegner and Miniami estimates \eqref{eq:wegner1} and \eqref{eq:minami1}, respectively, are valid for these random Schr\"odinger operators $H_\omega^{\ell,p}$ at negative energies.
We denote by $K_{\ell, p}(z, \omega)$, $K_{L}(z, \omega)$ and $K(z,\omega)$ the matrix Schr\"odinger operator defined in \eqref{eq:kernel1} for Hamiltonians $H_\omega^{\ell,p}$, $H_\omega^L$, and $H_\omega$, respectively.

The local eigenvalue statistics (LES) for each local Hamiltonian $H_\omega^{\ell,p}$, denoted by $\eta_\omega^{\ell,p}$, are independent point processes.
The collection $\{ \eta_\omega^{\ell,p} \}$ forms a \emph{uniformly asymptotically negligible array} ($uana$) of independent random point processes. This means that\mr
\beq\label{eq:uana1}
\lim_{L \rightarrow \infty} \sup_{p=1, \ldots, N_L} \Pp \Big( \eta_\omega^{\ell,p}(I)>\eta\Big) = 0.
\eeq
This follows from the Wegner estimate, Theorem \ref{thm:wegner1}, for the local Hamiltonians (see, for example, \cite{hk1,min1}).
We define the process $\zeta_\omega^L := \sum_p \eta_\omega^{\ell,p}$.

The proof of Theorem \ref{thm:main1} now consists of two main steps. In section \ref{sec:approx-pt-proc1},
we prove that the two local point processes $\xi_\omega^L$ and $\zeta_\omega^L$ have the same limit point showing that
\beq\label{eq:weak-convergence1}
\lim_{L \rightarrow \infty}  \E \{ \xi_\omega^L [f]  - \zeta_\omega^L [f] \} = 0,
\eeq
for appropriate test functions $f$. We will assume this result in this section and, following
\cite[section 6]{cgk2}, we prove that the local point process $\zeta_\omega^L$ associated with the $uana$ converges weakly to a Poisson point process.

In this section, we use arguments, as in \cite{min1} and \cite[chapter 11]{daley-vere-jones1},
adapted to the random delta-interaction model, to prove that the point process $\zeta_\omega^L$ associated with the $uana$ converges weakly to a Poisson point process with intensity measure $n(E_0) ds$, for $E_0 \in (-\infty, -\tilde{E_0}] \cap \Sigma^{\rm CL}$. Standard tightness bounds, using the Wegner estimate, establish the existence of the limiting point processes. We first establish the intensity of the limiting point process in section \ref{subsec:intensity1} following the arguments in \cite{cgk2}. We then use the Minami estimate, Theorem \ref{thm:minami-est1}, to show the nonexistence of double points in section \ref{subsec:no-double-pts1}. These results establish the uniqueness of the limit point and its Poisson nature.

In addition to the localization estimate on the expectation of small moments of the resolvents, Theorem \ref{prop:loc1}, the exponential decay bounds on the free Green's functions $G_0(x,y;z)$ and $G_0^L(x,y;z)$ play an essential role in this and the following section. At a negative energy $E_0 < 0$, the Green's functions exhibit the following behavior. There exists a constant $C_0 > 0$ and $c(E_0) = a_0 \sqrt{|E_0|}$, with $a_0$ independent of $L$ and $E_0$,  so that for $z \in \C \backslash [0, \infty)$, and $\Re z \in [E_0 - \epsilon, E_0 + \epsilon]$, for $\epsilon \geq 0$ small,
\beq\label{eq:green-exp-decay1}
| G_0^L (x,y;z)| , ~~|G_0(x,y;z)| \leq C_0 \frac{e^{- c(E_0)\|x-y\|}}{\|x-y\|} .
\eeq

\subsection{Intensity of the limiting point process}\label{subsec:intensity1}

We compute $\lim_{L \rightarrow \infty} \E \{  \zeta_\omega^L (I) \}$ that determines the intensity of the limiting process for the sequence $\zeta_\omega^L$. We proved in \cite[Corollary 7.5]{hkk1} that the density of states $n(E)$ exists for the random point interaction Schr\"odinger operators. Furthermore, the DOS
belongs to $L^1_{loc}(\R)$. These results follow from the Lipschitz continuity of the IDS proved in \cite{hkk1}.
The main result of this section is the following proposition.

\begin{prop}\label{prop:one-pt-est1}
For the $uana$ $\{ \eta_\omega^{\ell,p} \}$, any $E_0 \in (-\infty, \tilde{E_0}] \cap \Sigma^{\rm CL}$ for which $n(E_0) \neq 0$, and any 
interval $I \subset \R$, we have
\beq\label{eq:one-pt-est1}
\lim_{L \rightarrow \infty} \sum_{p=1}^{N_L} \Pp \{ \eta_\omega^{\ell,p}(I) = 1 \} = \lim_{L \rightarrow \infty} \E \{  \zeta_\omega^L (I)  \} =  n(E_0)|I| .
\eeq
\end{prop}

The first equality in Proposition \ref{prop:one-pt-est1} follows from the Minami estimate as in Proposition \ref{prop:two-pt-est1}.
Since the model is on $\R^d$, we follow the approach of \cite[section 6]{cgk1} in the proof of Proposition \ref{prop:one-pt-est1}. We define another measure $\Theta_\omega^\Lambda (\cdot)$ using a cut-off of the spectral projections of the infinite-volume Hamiltonian $H_\omega$ in the region of complete localization. Let $\chi_\Lambda$ be the characteristic function for the cube $\Lambda$. For any Borel $B \subset \R$, let $E_\omega (B)$ be the spectral projection for $H_\omega$ and $B$. We then define the measure $\Theta_\omega^\Lambda$ as
\beq\label{eq:theta1}
B \subset \R \rightarrow \Theta_\omega^\Lambda (B) :=  {\rm Tr} ~ \left\{ \chi_\Lambda E_\omega \left( E_0 + \frac{B}{|\Lambda|} \right) \chi_\Lambda \right\} .
\eeq
Because of translation covariance, it follows that
\beq\label{eq:theta-dos1}
\E \{ \Theta_\omega^\Lambda (B) \} = | \Lambda | \nu \left( E_0 + \frac{B}{|\Lambda|} \right) ,
\eeq
where $\nu$ is the density of states measure given by
\beq\label{eq:dos-m1}
\nu(B) = \E \{ {\rm Tr} \chi_0 E_\omega(B) \chi_0 \},
\eeq
where $\chi_0$ is the projection onto the unit cube in $\R^d$ centered at the origin.

From the proof of \cite[Lemma 1]{min1}, it suffices to prove \eqref{eq:one-pt-est1} for functions
in $f \in \mathcal{A}$ since they approximate characteristic functions on 
intervals $I \subset \R$.
The set $\mathcal{A}$ consists of all functions of the form
\beq\label{eq:special-class1}
f (x) = \sum_{i=1}^n a_i f_{\zeta_i}(x),
\eeq
for $n \geq 1$ and $a_i > 0$. For $\zeta = \sigma + i \tau$, we define $f_\zeta (x)$ to be
\beq\label{eq:test-fnc2}
f_\zeta (x) := \frac{\tau}{(x - \sigma)^2 + \tau^2}.
\eeq
In \eqref{eq:special-class1}, we take $\tau > 0$ and $\sigma_i \in \R$, for $i = 1, \ldots, n$.
The basic result of this section is the following proposition. The proof of Proposition \ref{prop:theta-approx1} is similar to the proof in section \ref{sec:approx-pt-proc1} of \eqref{eq:weak-convergence1}.

\begin{prop}\label{prop:theta-approx1}
For any function $f \in \mathcal{A}$, the local point process $\xi_\omega^{L}[f]$ associated with the local Hamiltonian $H_\omega^L$ and defined in \eqref{eq:les1} satisfies
\beq\label{eq:theta-approx1}
\lim_{L \rightarrow \infty} \E \{  \xi_\omega^L [f] - \Theta_\omega^L [f] \}  = 0 .
\eeq
\end{prop}

\begin{proof}
1. It suffices to establish \eqref{eq:theta-approx1} for functions $f_\zeta \in \mathcal{A}$, with $\zeta = \sigma + i \tau$, for $\tau > 0$ and $\sigma \in \R$, so we will prove that
\mr\beq\label{eq:theta-approx2}
\E \{ \xi_\omega^\Lambda [f_\zeta]- \Theta_\omega^\Lambda [f_\zeta] \} = \E \{ {\rm Tr} f_{\zeta}(|\Lambda| (H_\omega^\Lambda - E_0)) -
{\rm Tr} \chi_{\Lambda} f_{\zeta}(|\Lambda| (H_\omega - E_0)) \chi_\Lambda \}
\eeq
vanishes as $L \rightarrow \infty$. The term on the right in \eqref{eq:theta-approx2} may be reduced to the Green's function  at $z= E_0 + \frac{\sigma+ i \tau}{|\Lambda|}$. Since $E_0 < 0$, we have ${\rm dist}(z, [0, \infty))
\geq | |E_0| - \frac{|\sigma|}{|\Lambda_L|}|$, so $\| R_0^X (z)\|$ is bounded independent of $|\Lambda_L|$.
Using the explicit form of $f_\zeta$ in \eqref{eq:special-class1}, with $\zeta = \sigma+ i \tau$, we obtain
\bea\label{eq:theta-approx3}
\lefteqn{ \E \{ {\rm Tr} ~[ f(|\Lambda| (H_\omega^\Lambda - E_0)) -  \chi_\Lambda  f(|\Lambda|(H_\omega - E_0))\chi_\Lambda  ] \} } \nonumber \\
&=& \frac{1}{|\Lambda|} \E ~ \{ {\rm Tr} \Im \chi_{{\Lambda}_L}  [ R_\omega^{\Lambda_L}  (z) -   R_\omega  (z)  ] \chi_{{\Lambda}_L} \}
           \nonumber \\
 &=& \frac{1}{|\Lambda|} {\rm Tr} \Im  \chi_{{\Lambda}_L} [ R_0  (z) -   R_0^{\Lambda_L}  (z)  ] \chi_{{\Lambda}_L}  \nonumber \\
 & &  + \frac{1}{|\Lambda|} \sum_{(k,m) \in (\Z^3)^2}  \E\Big\{ {\rm Tr} ~\Im \Big( \chi_{{\Lambda}_L}  [ R_0^{\Lambda_L}(z) P_k [K_L^{-1} (z, \omega)] P_m R_0^{\Lambda_L}(z)   \nonumber  \\
 & &  -  R_0 (z) P_k [K^{-1} (z, \omega)]_{km} P_m R_0 (z) ] \chi_{{\Lambda}_L}\Big) \Big\}  ,
\eea
where $[K_L^{-1}(z,\omega)]_{km} = 0$ if $k ~{\rm or}~ m \in \Z^3 \backslash \tilde{\Lambda}_L$ and $P_j$ stands for evaluating the operator kernels (resp. the matrices) to the left and to the right at the point $j$.
We set
\beq\label{eq:theta-AL}
A_L := \frac{1}{|\Lambda|} {\rm Tr} \Im \chi_{{\Lambda}_L} \left[ R_0^{\Lambda_L}  (z) -   R_0  (z)  \right] \chi_{{\Lambda}_L} ,
\eeq
which is deterministic, and
\bea\label{eq:theta-BL}
B_L & := & \frac{1}{|\Lambda|} \sum_{(k,m) \in {(\Z^3)^2}}  \E \left\{ {\rm Tr} \Im \chi_{{\Lambda}_L} [ R_0^{\Lambda_L}(z) P_k [K_L^{-1}(z,\omega) ]_{km} P_m R_0^{\Lambda_L}(z) \right.  \nonumber \\
 & &  \left. - R_0 (z) P_k [K^{-1}(z,\omega) ]_{km} P_m R_0 (z) ] \chi_{{\Lambda}_L} \right\} .
\eea

\noindent
2. To control the deterministic term $A_L$ in \eqref{eq:theta-AL}, we introduce a cut-off function on the subset of $\Lambda_L$ defined by ${\Lambda}^\prime_L := \{ x \in \Lambda_L ~|~ {\rm dist} ~(x, \partial \Lambda_L ) > C \log L \}$. We denote by $\partial {\Lambda}^\prime_L := \Lambda_L \backslash {\Lambda}^\prime_L$.
We write characteristic functions on these sets as
$\chi_{\Lambda_L} =  \chi_{{\Lambda}^\prime_L}  + \chi_{\partial {\Lambda}^\prime_L}$.
Inserting this decomposition into $A_L$ in \eqref{eq:theta-approx3}, we obtain two terms, one localized in ${\Lambda}^\prime_L$:
\beq\label{eq:theta-approx4}
\frac{1}{|\Lambda_L|}  \left| {\rm Tr} \chi_{{\Lambda}^\prime_L} \left( R_0  (z)
-   R_0^{\Lambda_L}  (z)  \right) \right| ,
\eeq
and one localized in the boundary $\partial {\Lambda}^\prime_L$:
\beq\label{eq:theta-approx5}
\frac{1}{|\Lambda_L|} \left| {\rm Tr} \chi_{\partial {\Lambda}^\prime_L} \left( R_0  (z)
-   R_0^{\Lambda_L}  (z)  \right) \right| .
\eeq
We estimate \eqref{eq:theta-approx4} using the fact that $(H_0 - H_0^L)g = 0$ if $g$ is a smooth function supported in $\Lambda_L$ away from $\partial \Lambda_L$. That is, the difference between $H_0$ and $H_0^L$  is localized near $\partial \Lambda_L$ whereas the trace in  \eqref{eq:theta-approx4}  is localized to ${\Lambda}^\prime_L$ and the Green's functions decay exponentially \eqref{eq:green-exp-decay1} since $\Re z < 0$. We estimate \eqref{eq:theta-approx5} using the fact that $|\partial {\Lambda}^\prime_L|$ is small relative to $|\Lambda_L|$.
We introduce a boundary operator $\Gamma_L$ described by an application of 
 Green's Theorem so that $\Gamma_L R_0^L(z)$ is the restriction of $\nu \cdot \nabla R_0^L(z)$ to $\partial \Lambda_L$, where $\nu$ is the outward normal unit vector.
Let $\varphi_L$ be a smooth function localized near $\partial \Lambda_L$ so that $\Gamma_L = \varphi_L \Gamma_L$.
Writing $Y_L$ for ${\Lambda}^\prime_L$ or $\partial {\Lambda}^\prime_L$, we rewrite the traces in \eqref{eq:theta-approx4} and in \eqref{eq:theta-approx5} using this resolvent formula in a compact notation:
\bea\label{eq:theta-approx6}
| {\rm Tr} ~\{ \chi_{Y_L} (  R_0 (z) - R_0^{\Lambda_L}(z)) \chi_{Y_L}  \} | & = & | {\rm Tr} ~\{ \chi_{Y_L}  R_0^{}(z) \varphi_L \Gamma_L R_0^{\Lambda_L}(z)  \chi_{{Y}_L} \} | \nonumber \\
 & \leq & \| \chi_{Y_L}  R_0^{}(z) \varphi_L \Gamma_L R_0^{\Lambda_L}(z) \chi_{Y_L} \|_1 \nonumber \\
  & \leq & \| \chi_{Y_L}  R_0^{}(z) \varphi_{L} \|_2
 \| \varphi_L  \Gamma_L R_0^{\Lambda_L}(z) \chi_{Y_L}\|_2 .
 \eea
For \eqref{eq:theta-approx4}, we have $Y_L = {\Lambda}^\prime_L$ and note that
\beq\label{eq:theta-approx6a}
 \| \chi_{{\Lambda}^\prime_L}  R_0^{}(z) \varphi_{L} \|_2 \leq C [L^2 \log L]^\frac{1}{2},
 \eeq
 using the exponential decay of the Green's function \eqref{eq:green-exp-decay1} to control the integral over ${\Lambda}^\prime_L$.
 A similar estimate holds for the second term since the gradient of the Green's function is integrable and decays exponentially. As a consequence, we have the bound
  \beq\label{eq:theta-approx4a}
\frac{1}{|\Lambda_L|}  \left| {\rm Tr} \chi_{{\Lambda}^\prime_L} \left( R_0  (z)
-   R_0^{\Lambda_L}  (z)  \right) \right| \leq C_0 \frac{\log L}{L} .
\eeq
Concerning \eqref{eq:theta-approx5} with $Y_L = \partial{\Lambda}^\prime_L$, we have the bound
\beq\label{eq:theta-approx6b}
 \| \chi_{\partial {\Lambda}^\prime_L}  R_0^{}(z) \varphi_{L} \|_2 \leq C [L^2 \log L]^\frac{1}{2},
 \eeq
and similarly for the second factor on the right in \eqref{eq:theta-approx6}, so that
\beq\label{eq:theta-approx5a}
\frac{1}{|\Lambda_L|}  \left| {\rm Tr} \chi_{\partial {\Lambda}^\prime_L} \left( R_0  (z)
-   R_0^{\Lambda_L}  (z)  \right) \right| \leq C_0 \frac{\log L}{L} .
\eeq
Bounds \eqref{eq:theta-approx4a} and \eqref{eq:theta-approx5a} show that $A_L$ vanishes as $L \rightarrow \infty$.

\noindent
3. The term $B_L$ in \eqref{eq:theta-BL} consists of operators with singular, but square integrable, integral kernels with singularities located at lattice points in $\Lambda_L$. We estimate this term as in section \ref{subsec:second-term1} by decomposing $\Lambda_L = M_{L,\epsilon} \cup M_{L, \epsilon}^c$, where we define $M_{L, \epsilon} := \left( \cup_{j \in \tilde{\Lambda}_L} B(j,\epsilon) \right)$, with $B(j, \epsilon)$ denoting the ball centered at $j$ of radius $\epsilon$.  As in Lemma \ref{lemma:second-term-singular1}, the contribution to the trace in $B_L$ term coming from $M_{L, \epsilon}$ vanishes in the limit $L \rightarrow \infty$, provided $\epsilon = o (|\Lambda_L|^{-2})$. Unlike in section \ref{sec:approx-pt-proc1}, we do not take the $s^{th}$-power and, consequently, we do not use the localization bound on $\E \{[ K_L^{-1}(z,\omega)]_{km}|^s\}$,  but the bounds following from spectral averaging in \eqref{eq:bdd-Kexpectation1}.  We denote the nonsingular term of $B_L$, obtained by integrating over $\Lambda_L \cap M_{L, \epsilon}^c$, by $B_L^c$. We first separate the sum over $(k,m) \in (\Z^3)^2$ into two terms: A main term for which $(k,m) \in (\tilde{\Lambda}_L)^2$, denoted by $\tilde{B}_L^c$,  and the rest for which at least $k$ or $m$ is not in $\tilde{\Lambda}_L$. The distinguishing feature of the latter terms is the absence of the kernel $[K_L^{-1}(z,\omega)]_{km}$. We write $B_L^c =  B_L^{(1)} +  B_L^{(2)} + B_L^{(3)}$, where
\beq\label{eq:theta-BL1}
B_L^{(1)} := \frac{1}{|\Lambda|} \sum_{(k,m) \in {(\Z^3 \backslash \tilde{\Lambda}_L)^2}}  \E \left\{ [K^{-1}(z, \omega )]_{km} \right\} \int_{\Lambda_L \cap  M_{L, \epsilon}^c}  ~  | G_0^\Lambda (x,k;z)| | G_0 (x,m;z)| ~d^3x,
\eeq
and a term of the form
\beq\label{eq:theta-BL2}
B_L^{(2)} := \frac{1}{|\Lambda|} \sum_{\stackrel{m \in {\Z^3 \backslash \tilde{\Lambda}_L}}{k \in \tilde{\Lambda}_L}}   \E \left\{ [K^{-1}(z, \omega )]_{km} \right\} \int_{\Lambda_L \cap  M_{L, \epsilon}^c}  ~  | G_0^\Lambda (x,k;z)| | G_0 (x,m;z)| ~d^3x .
\eeq
Note that if $k\in \Z^3 \backslash \tilde{\Lambda}_L$, then $| G_0^\Lambda (x,k;z)|=0$ so that $B_L^{(1)} = 0$.
We use the spectral averaging result \cite[(5.7)]{hkk1} that implies
\beq\label{eq:bdd-Kexpectation1}
\E \left\{ [K_L^{-1}(z, \omega )]_{km} \right\}, ~ \E \left\{ [K^{-1}(z, \omega )]_{km} \right\}  \leq D(E_0),
\eeq
uniform in compact subsets of $\C \backslash [0, \infty)$.

\noindent
4. To estimate $B_L^{(2)}$, we have $\| x-m \| > \frac{1}{2}$, so, using \eqref{eq:bdd-Kexpectation1}, we have the bound
\beq\label{eq:theta-BL2-1}
B_L^{(2)} \leq  \frac{D(E_0)}{|\Lambda|} \sum_{\stackrel{m \in {\Z^3 \backslash \tilde{\Lambda}_L}}{k \in \tilde{\Lambda}_L}}   \int_{\Lambda_L \cap  M_{L, \epsilon}^c}  ~\frac{e^{- c(E_0) \|x-k\|} }{\|x-k\|} e^{- c(E_0) \|x-m\|} ~d^3x .
\eeq
We divide the region $\Lambda$ into $\Lambda_L^\prime := \{ x \in \Lambda_L ~|~  {\rm dist}(x, \partial \Lambda_L) > C\log L \}$, and the boundary $\partial \Lambda_L^\prime$. It is easy to see that contribution to $B_L^{(2)}$
from the integral over $\Lambda_L^\prime$ decays like $L^{- C c(E_0) }$. As for the
the contribution from the integral over $\partial \Lambda_L^\prime$,
it follows from \eqref{eq:theta-BL2-1} that
\bea\label{eq:theta-BL2-2}
\lefteqn{ \frac{D(E_0)}{|\Lambda|} \sum_{k \in \tilde{\Lambda}_L}   \int_{\partial \Lambda_L^\prime \cap  M_{L, \epsilon}^c}  ~\frac{e^{- c(E_0) \|x-k\|} }{\|x-k\|} ~\sum_ {m \in {\Z^3 \backslash \tilde{\Lambda}_L}} e^{- c(E_0) \|x-m\|} ~d^3x} \nonumber \\
  &\leq & \frac{D(E_0)}{|\Lambda|} \sum_{k \in \tilde{\Lambda}_L}   \int_{\partial \Lambda_L^\prime \cap  M_{L, \epsilon}^c} ~ \frac{e^{- c(E_0) \|x-k\|} }{\|x-k\|} ~d^3 x \nonumber \\
  & \leq & \frac{D(E_0)\log L}{L},
      \eea
as follows from dividing the $k$-sum into $k \in \partial \Lambda_L^\prime$ and $k \in {\tilde \Lambda}_L^\prime$.

\noindent
5. It remains to estimate the main term  $B_L^{(3)}$ in ${B}_L^c$ involving the sum $(k,m) \in \tilde{\Lambda}_L^2$.
Since the integral is over $\Lambda_L \cap M_{L, \epsilon}^c$, the kernels are continuous so the trace in \eqref{eq:theta-BL} may be written as
\bea\label{eq:theta-approx10}
{B}_L^{(3)}  & : = & \frac{1}{|\Lambda|}
  \sum_{(k,m) \in \tilde{\Lambda}_L^2} \int_{\Lambda_L \cap M_{L, \epsilon}^c} ~d^3x ~ \E [ G_0^\Lambda (x,k;z)[K_L(z, \omega)^{-1} - K (z, \omega)^{-1} ]_{km}
  G_0 (m,x;z)] \nonumber \\
 & &  + \mathcal{E}_1 (L) + \mathcal{E}_2(L) \nonumber \\
 & = & \tilde{B}_L^c + \mathcal{E}_1 (L) + \mathcal{E}_2(L) ,
\eea
where the error terms are
\beq\label{eq:theta-error1}
\mathcal{E}_1 (L) := \frac{1}{|\Lambda|}
  \sum_{(k,m) \in \tilde{\Lambda}_L^2} \int_{\Lambda_L \cap M_{L, \epsilon}^c} ~d^3x ~ \E [ G_0^\Lambda (x,k;z)[K_L (z, \omega)^{-1}]_{km}
  [G_0^\Lambda (m,x;z)] - G_0 (m,x;z)] ,
  \eeq
and
\beq\label{eq:theta-error2}
\mathcal{E}_2(L) := \frac{1}{|\Lambda|}
  \sum_{(k,m) \in \tilde{\Lambda}_L^2} \int_{\Lambda_L \cap M_{L, \epsilon}^c} ~d^3x ~ \E [ G_0^\Lambda (x,k;z) - G_0(x,k;z)][K(z, \omega)^{-1}]_{km}
  G_0 (m,x;z).
\eeq

\noindent
6. We treat the error term $\mathcal{E}_1 (L)$. The other term $\mathcal{E}_2 (L)$ may be dealt with in a similar manner. We consider three concentric cubes $\Lambda_{L''} \subset \Lambda_{L'} \subset \Lambda_L$,
with $\mbox{dist} ~( \Lambda_{L^\prime}, \partial \Lambda_L) = \ell^\beta$ and $\mbox{dist} ~( \Lambda_{L''}, \partial \Lambda_{L^\prime}) = \ell^{\beta'}$,  for $0 < \beta, \beta' < 1$.
Let $\chi_{\Lambda_{L^\prime}}$ be a cut-off function with $\chi_{\Lambda_{L^\prime}} \chi_{\Lambda_L} = \chi_{\Lambda_{L^\prime}}$. Similarly, we take another cut-off function $\chi_{\Lambda_{L''}}$ supported on ${\Lambda_{L''}}$ with
$\chi_{\Lambda_{L^\prime}} \chi_{\Lambda_{L''}} = \chi_{\Lambda_{L''}}$.
By the geometric resolvent identity, we have
\beq\label{eq:gre1}
\chi_{\Lambda_{L^\prime}} R_0^\Lambda (z) =  R_0 (z)\chi_{\Lambda_{L^\prime}} +  R_0 (z) C_{\partial\Lambda_{L^\prime}} R^\Lambda_0 (z).
\eeq
Multiplying both sides of \eqref{eq:gre1} by the cut-off function $\chi_{\Lambda_{L''}}$,
we obtain from \eqref{eq:gre1}:
\beq\label{eq:gre2}
\chi_{\Lambda_{L''}} R_0^\Lambda (z) \chi_{\Lambda_{L''}} = \chi_{\Lambda_{L''}} R_0 (z) \chi_{\Lambda_{L''}} +
\chi_{\Lambda_{L''}} R_0 (z) C_{\partial\Lambda_{L^\prime}} R^\Lambda_0 (z) \chi_{\Lambda_{L''}}.
\eeq
We repeated use the following bound. If $X \geq 0$ is a non-negative random variable depending on $\omega_j, j \in \tilde{\Lambda}_L$, then
\beq\label{eq:expectation1}
\E  [X]  =  \E [ X^{1-s} X^s]  \leq  \E [  X^{2(1-s)}]^{\frac{1}{2}} ~ \E [ X^{2s} ]^{\frac{1}{2}}.
\eeq
We substitute the functions $1 = \chi_{\Lambda_{L'}} + (1 - \chi_{\Lambda_{L'}})$ into the integral of \eqref{eq:theta-error1} and obtain two terms:
$\mathcal{E}_1 (L) = \mathcal{E}_1 (L; L^\prime)  + \mathcal{E}_1 (L, \partial L^\prime)$. We first treat $\mathcal{E}_1 (L, \partial L^\prime)$. For this, we apply \eqref{eq:expectation1} with $X := |[K_L(z,\omega)^{-1}]_{km}|$ and note that since $\Im z = \frac{\tau}{|\Lambda_L|}$,
we have $\E [  X^{2(1-s)}]^{\frac{1}{2}} = |\Lambda_L|^{1-s}$. Then, the exponential decay of the Green's functions \eqref{eq:green-exp-decay1}, together with the localization bound \eqref{eq:K-exp-decay1} and inequality \eqref{eq:expectation1}, indicate that
\bea\label{eq:theta1-bound1}
|\mathcal{E}_1 (L, \partial L^\prime)| & \leq & \frac{1}{|\Lambda|^s}
  \sum_{(k,m) \in \tilde{\Lambda}_L^2} e^{- \gamma_{s,3} (z) \| k-m\|} \nonumber \\
   & & \times  \int_{\Lambda_{L} \cap M_{L, \epsilon}^c} ~d^3x ~ (1 - \chi_{\Lambda_{L'}}(x))
   | G_0^\Lambda (x,k;z) ~ [G_0^\Lambda (m,x;z)C_{\partial\Lambda_{L^\prime}} G_0 (m,x;z)]| \nonumber \\
& \leq & C \left( \frac{\log L}{L^{3s}} \right).
\eea
As for the term $\mathcal{E}_1 (L; L^\prime)$, we again separate it into two terms using the function $\chi_{\Lambda_{L''}}$:
$1 = \chi_{\Lambda_{L''}} + (1-\chi_{\Lambda_{L''}})$. The term involving $(1-\chi_{\Lambda_{L''}})$ is treated as in \eqref{eq:theta1-bound1}.
As for the first term involving $\chi_{\Lambda_{L''}}$, again using the localization estimate \eqref{eq:K-exp-decay1} and \eqref{eq:expectation1}, we obtain,
\bea\label{eq:theta1-bound2}
\lefteqn{\frac{1}{|\Lambda|}  \sum_{(k,m) \in \tilde{\Lambda}_L^2} \int_{\Lambda_L \cap M_{L, \epsilon}^c} ~d^3x ~ \chi_{\Lambda_{L''}}(x) ~ | G_0^\Lambda (x,k;z) \E \{ [K_L (z, \omega)^{-1}]_{km} \}  [G_0^\Lambda (m,x;z) - G_0 (m,x;z)]|} \nonumber \\
 & \leq & \frac{1}{|\Lambda|^s}  \sum_{(k,m) \in \tilde{\Lambda}_L^2} e^{- \gamma_{s,3} (z) \| k-m\|} \int_{\Lambda_L \cap M_{L, \epsilon}^c} ~d^3x ~ \frac{e^{- c(E_0) \|x-k\|}}{\| x-k\|} \chi_{\Lambda_{L''}}(x) \frac{e^{- c(E_0) \|x-k\|}}{\| x-k\|}   C({\partial\Lambda_{L^\prime}})(x)_{km} \nonumber \\
  & & \times \frac{e^{- c(E_0) \|x-m\|}}{\| x-m\|} \chi_{\Lambda_{L''}} (x) .
 \eea
Since ${\rm dist}~( \Lambda_{L''}, \partial \Lambda_{L^\prime} ) = \ell^\beta$, the Green's functions are bounded by $e^{- c(E_0) \ell^\beta}$. The local singularities are no worse that $\log L$.

\noindent
7. Returning to the main term $\tilde{B}_L^c$ in \eqref{eq:theta-approx10}, we will estimate $\E \{ |\tilde{B}_L^c|^\frac{s}{2} \}$, for $s \in (0, 1)$:
\beq\label{eq:theta-approx11}
\E \{ |\tilde{B}_L^c|^\frac{s}{2} \} = \frac{1}{|\Lambda|^\frac{s}{2}}
  \sum_{(k,m) \in \tilde{\Lambda}_L^2} \int_{\Lambda_L \cap M_{L, \epsilon}^c} ~d^3x ~ | G_0^\Lambda (x,k;z)|^\frac{s}{2} \E \{| [K_L(z, \omega)^{-1} - K (z, \omega)^{-1} ]_{km} |^\frac{s}{2} \} |G_0 (m,x;z)|^\frac{s}{2} .
\eeq
We apply the second resolvent equation to the matrices $K_L(z,\omega)$ and $K (z,\omega)$ defined in \eqref{eq:K-kernel1}. Defining the difference $M^L_{nr}(z) :=  [ K_L(z,\omega) - K (z,\omega) ]_{nr}$, we
have:
\beq\label{eq:sec-resolv-K0}
[ K_L(z,\omega)^{-1} - K(z,\omega)^{-1} ]_{km} =   \sum_{(n,r) \in (\Z^3)^2} M_{nr}(z)   ~[K_L(z,\omega)^{-1}]_{kn} ~[K(z,\omega)^{-1}]_{rm} .
\eeq
According to \eqref{eq:theta-approx11}, if $n$ in the sum in \eqref{eq:sec-resolv-K0} is not in $\tilde{\Lambda}_L$,
$M^L_{nr}(z) = 0$, so the sum in \eqref{eq:sec-resolv-K0} has two terms: (1) We define the set $\tilde{X}_1 := \{ (n,r) \in \tilde{\Lambda}_L^2 ~|~ n \neq r \}$, for which, 
\beq\label{eq:K-Lambda-offDiag1}
  M^L_{nr;1}(z) = [G_0(n,r;z) - G_0^L(n,r;z)] , ~~~ (n,r) \in \tilde{X}_1,
\eeq
and (2) the set $\tilde{X}_2 := \{ (n,r) ~|~ n \in \tilde{\Lambda}_L, r \in \Z^3 \backslash \tilde{\Lambda}_L \}$, so that necessarily $n \neq r$, for which
\beq\label{eq:K-Lambda-Ext1}
M^L_{nr;2}(z) = G_0(n,r;z), ~~~ (n,r) \in \tilde{X}_2.
\eeq
By Green's theorem, we can write $M_{nr;1}^L(z)$ in case (1) as
\beq\label{eq:green33}
M_{nr;1}^L(z) = \int_{\partial \Lambda_L} dS(w) G_0 (n, w;z) (\partial_\nu G_0^L)(w,r;z) , ~~~ (n,r) \in \tilde{X}_1.
\eeq
Since $\| w - n\| > \frac{1}{2}$ and $\|w-n \| \geq {\rm dist}(n, \partial \Lambda_L)$, we obtain the bound
on $M_{nr;1}^L(z)$ via exponential decay and integrating \eqref{eq:green33}:
\bea\label{eq:M-decay1}
|M_{nr;1}^L(z)| &= &  M_0 e^{- \frac{c(E_0)}{2}[{\rm dist}(r, \partial \Lambda_L) + {\rm dist}(n, \partial \Lambda_L) ]} \int_{\partial \Lambda_L} dS(w) e^{- \frac{c(E_0)}{2} [\| r-w\| + \|n-w\|]} \nonumber \\
 & \leq & M_0 e^{- \frac{c(E_0)}{2}[{\rm dist}(r, \partial \Lambda_L) + {\rm dist}(n, \partial \Lambda_L) ]} ,
 \eea
for $(n,r) \in \tilde{X}_1$, and a constant $M_0>0$ uniform in $(n,r)$ and in $L>0$.
Similarly, we obtain for $ M_{nr;2}^L(z)$:
\beq\label{eq:M-decay2}
| M_{nr;2}^L(z)|  \leq   e^{- c(E_0){\rm dist}(r, \partial \Lambda_L)},
\eeq
for $ (n,r) \in \tilde{X}_2$.
Hence, the terms $M_{nr}^L(z)$, in cases (1) and (2) occurring in \eqref{eq:theta-approx11}, are
nonrandom and decay exponentially away from $\partial \Lambda_L$ because $E_0 < 0$.
Inserting \eqref{eq:sec-resolv-K0} into the integral in \eqref{eq:theta-approx6}, we obtain the bound
\bea\label{eq:theta-approx12}
\E \{ \tilde{|B}_L^c|^\frac{s}{2} \} & \leq & {  \frac{1}{|\Lambda_L|^\frac{s}{2}}  \sum_{(k,m) \in \tilde{\Lambda}_L^2}
    \sum_{(n,r) \in \tilde{X}_1 \cup \tilde{X}_2} ~M_{nr}(z)^\frac{s}{2} ~ \E \{ |(K_L(z, \omega)^{-1})_{kn}|^{s}\}^{\frac{1}{2}} \E \{ |(K (z, \omega)^{-1})_{rm}|^{s} \}^{\frac{1}{2}} } \nonumber \\
  & & \times        \int_{\Lambda_L \cap M_\epsilon^c} ~  |G_0^L (x,k;z) G_0 (m,x;z)|^\frac{s}{2}  ~d^3x ,
 \eea
where $\tilde{X}_1 := \tilde{\Lambda}_L^2$ and $\tilde{X}_2 := \{ (n,r) ~|~ n \in \tilde{\Lambda}_L, r \in \Z^3 \backslash \tilde{\Lambda}_L \}$.
The localization estimate \eqref{eq:K-exp-decay1} for $0 < s < 1$ provides the bound
\beq\label{eq:K-exp-decay21}
\E \{ |(K_L(z, \omega)^{-1})_{kn}|^{s}\}^{\frac{1}{2}} \E \{ |(K(z, \omega)^{-1})_{rm}|^{s} \}^{\frac{1}{2}}
 \leq C_s^2 e^{-s \gamma_{s,3} (z) ( \|k-n \| + \|r-m\| )}.
 \eeq
By standard exponential decay estimates (\ref{eq:green-exp-decay1}) on the local Green's functions we have
\beq\label{eq:greens-fnc-decay11}
\int_{\Lambda_L \cap M_\epsilon^c} ~  |G_0^L (x,k;z) G_0 (m,x;z)|  ~d^3x
 \leq \int_{\Lambda_L \cap M_\epsilon^c} ~  \frac{e^{- c(E_0)  \| k-x\|}}{\| k-x\|}  \frac{e^{- c(E_0) \| m-x\|}}{\| m-x\|} ~d^3x.
\eeq
Because the kernel is locally integrable we find that
\beq\label{eq:greens-fnc-decay21}
\int_{\Lambda_L \cap M_\epsilon^c} ~  |G_0^L (x,k;z) G_0 (m,x;z)|  ~d^3x
\leq
e^{- \tilde{c}(E_0)  \| k-m\|} ,
\eeq
where$\tilde{c}(E_0)$ differs from $c(E_0)$ be a constant.
Returning to \eqref{eq:theta-approx12}, we use the bound \eqref{eq:M-decay1} to obtain
\beq\label{eq:theta-approx13}
\frac{1}{L^\frac{3s}{2}} \sum_{(k,m) \in \tilde{\Lambda}_L^2}  \sum_{(n,r) \in \tilde{X}_1 } ~ e^{- {\frac{s c(E_0)}{2} [{\rm dist}(r, \partial \Lambda_L) + {\rm dist}(n, \partial \Lambda_L) ]}}  e^{-s \gamma_{s,3} (z) ( \|k-n \| + \|r-m\| )} e^{- \tilde{c}(E_0)  \| k-m\|} \leq \frac{C_1}{L^\frac{3s}{2}} ,
\eeq
and \eqref{eq:M-decay2} to obtain
\beq\label{eq:theta-approx14}
\frac{1}{L^\frac{3s}{2}}  \sum_{(k,m) \in \tilde{\Lambda}_L^2}
    \sum_{(n,r) \in \tilde{X}_2} ~e^{- \frac{s c(E_0)}{2} {\rm dist}(r, \partial \Lambda_L)} ~  e^{-s \gamma_{s,3} (z) ( \|k-n \| + \|r-m\| )}
    e^{- \tilde{c}_z  \| k-m\|} \leq \frac{C_2}{L^\frac{3s}{2}}.
\eeq
These estimates \eqref{eq:theta-approx13}--\eqref{eq:theta-approx14} show that $\tilde{B}_L^c$
vanishes in probability as $L \rightarrow \infty$, where $\tilde{B}_L^c$ is defined in \eqref{eq:theta-approx10}. This, together with the estimates on $B_L^{(2)}$ and the error terms ${\mathcal{E}}_j(L)$, for $j=1,2$, proves the proposition.
\end{proof}

Given this technical estimate, and the result \eqref{eq:resest4} relating $\xi_\omega^L$ to $\zeta_\omega^L$,
the proof of Proposition \ref{prop:one-pt-est1} follows the proof in \cite[section 6]{cgk2}. For any $f \in \mathcal{A}$, by Propositions \ref{prop:theta-approx1} and \ref{prop:local-uana-conv1}, we have
\beq\label{eq:theta-dos2}
\lim_{L \rightarrow \infty} \E \{ \zeta_\omega^{\Lambda_L} [f] \} =  \lim_{L \rightarrow \infty} \E \{ \Theta_\omega^{\Lambda_L} [f] \}.
\eeq
 To evaluate the last limit on the right in \eqref{eq:theta-dos2}, we use the Lebesgue differentiation theorem.
 \bea\label{eq:theta-dos3}
 \lim_{L \rightarrow \infty} \E \{ \Theta_\omega^{\Lambda_L} (I) \} & = & \lim_{L \rightarrow \infty} | \Lambda | \nu
 \left( E_0 + \frac{I}{|\Lambda|} \right) \nonumber \\
  & = & \lim_{L \rightarrow \infty} | \Lambda | \int_{E_0 + I |\Lambda_L|^{-1}} n(s) ~ds .
 \eea
 Since $E_0$ is a Lebesgue point of the DOS $n(s)$,
 we obtain
 \beq\label{eq:dos1}
  \lim_{L \rightarrow \infty} | \Lambda | \int_{E_0 + I |\Lambda_L|^{-1}} n(s) ~ds = n(E_0) |I| ,
  \eeq
verifying Proposition \ref{prop:one-pt-est1} establishing the intensity of the limiting point process of the uana.

\subsection{Elimination of double points}\label{subsec:no-double-pts1}

The second condition on the convergence of the $uana$ $\{ \eta_\omega^{\ell,p}\}$ guarantees that the limit process is a simple point process (see, for example, \cite[Proposition 11.1.IX]{daley-vere-jones1}).
The proof relies on the Minami estimate \eqref{eq:minami1} for length scale $\ell = L^\alpha$, with $0 < \alpha < 1$.

\begin{prop}\label{prop:two-pt-est1}
For the $uana$ $\{ \eta_\omega^{\ell,p} \}$, we have
\beq\label{eq:two-pt-est1}
\lim_{L \rightarrow \infty} \sum_{p=1}^{N_L} \Pp \{ \eta_\omega^{\ell,p}(I) \geq 2 \} = 0.
\eeq
\end{prop}

\begin{proof}
We use the Minami estimate \eqref{eq:minami1} on length scale $\ell$ for the local Hamiltonian $H_\omega^{\ell, p}$:
\bea\label{eq:two-pt-est2}
\sum_{p=1}^{N_L} \Pp \{ \eta_\omega^{\ell,p}(I) \geq 2 \} & \leq & \sum_{p=1}^{N_L} \Pp \{ \eta_\omega^{\ell,p}(I)[ \eta_\omega^{\ell,p}(I)-1] \geq 1 \} \nonumber \\
 & \leq & \sum_{p=1}^{N_L} \E \{ \eta_\omega^{\ell,p}(I)[ \eta_\omega^{\ell,p}(I)-1] \} \nonumber \\
  & \leq & \left( \frac{L}{\ell} \right)^d C_M \ell^{2d} \left( \frac{|I|}{L^d}  \right)^2 \nonumber \\
  & \leq & C_M |I|^2 \left( \frac{\ell}{L} \right)^d = C_M L^{-d(1 - \alpha)} ,
  \eea
and this vanishes in the limit $L \rightarrow \infty$.
\end{proof}

\subsection{LES for the $uana$}\label{subsec:les1}

The main results of sections \ref{subsec:intensity1} and \ref{subsec:no-double-pts1}
imply the following characterization of the limiting point process associated with the $uanu$.

\begin{theorem}\label{thm:uanaP1}
For any $E_0 \in (-\infty, \tilde{E_0}] \cap \Sigma^{\rm CL}$, so that $n(E_0) \neq 0$,  the process $\zeta_\omega^L$, constructed from the $uana$ $\{ \eta_\omega^{\ell,p} \}$, converges weakly to a Poisson point process with intensity measure $n(E_0) ds$.
\end{theorem}



\section{Approximation of the point process $\xi_\omega$ by a $uana$}\label{sec:approx-pt-proc1}
\setcounter{equation}{0}


The main result of section \ref{sec:les-indep-arrays1} is the convergence of the point process associated with the $uana$ $\zeta_\omega^L = \sum_{p=1}^{N_L} \eta_\omega^{\ell,p}$ to a Poisson point process $\xi^P$. That is, for all test functions $f \in C_0^+(\R)$, nonnegative, continuous functions of compact support, we have
\beq\label{eq:laplacetransf1} 
\lim_{L \rightarrow \infty} \EE \left\{ e^{-  \zeta_\omega^L [f]} \right\} =  \E \left\{ e^{- \xi^P [f]} \right\} ,
\eeq
where the right hand side is the characteristic function of the Poisson point process with intensity measure $n(E_0) ds$.
The local eigenvalue point process is centered at any energy $E_0 \in (- \infty, \tilde{E}_0) \cap \Sigma^{\rm CL}$, where localization has been proven in \cite{hkk1}.
In this section, we complete the proof of Theorem \ref{thm:main1}
by showing that $\xi_\omega^L$ has the same limit point as $\zeta_\omega^L$, which is,  by Theorem \ref{thm:uanaP1}, 
the Poisson point process with intensity $n(E_0) ds$.


Minami \cite[Lemma 1]{min1} provided a criteria for determining when a sequence of point processes $\xi_n$ converges weakly to a point process $\xi$. He proved that if the densities of the intensity measures of the $\xi_n$ are uniformly bounded and the density of the limiting process $\xi$ is also bounded by the same constant, then the weak convergence of the $\xi_n$ to $\xi$ is equivalent to the convergence of the Laplace transforms of $\xi_n [f]$ to $\xi [f]$ for all $f \in \mathcal{A}$.
In light of Theorem \ref{thm:uanaP1} and \eqref{eq:laplacetransf1}, it suffices to prove that
\beq\label{eq:laplacetransf3}
\lim_{L \rightarrow \infty} \EE\left\{  e^{-  \xi_\omega^L [f] } -  e^{- \zeta_\omega^L [f]} \right\} =0 .
\eeq
For any nonegative functions $X,Y > 0$ and for any $0 \leq  s \leq 1$, we have
\beq\label{eq:exp11}
|e^{-X} - e^{-Y}| \leq 2^{1-s} | e^{-X} - e^{-Y}|^s.
\eeq
Since we have the upper bound
\beq\label{eq:exp1}
| e^{-X} - e^{-Y}|^s =  \left|  \int_0^1  e^{-tX}(Y-X) e^{-(1-t)Y} ~dt \right|^s  \leq |X-Y|^s,
\eeq
the vanishing of the limit in \eqref{eq:laplacetransf3} is guaranteed by
\beq\label{eq:resest4}
\lim_{L \rightarrow \infty} \E \left\{ | \xi_\omega^L [f] - \zeta_\omega^L [f] |^s \right\} = 0,
\eeq
for $0 < s < 1$.
The proof of \eqref{eq:resest4}, that depends on the localization estimates and the exponential decay of the Green's functions, is the main result of this section.

\begin{prop}\label{prop:local-uana-conv1}
For the point process $\xi_\omega^L$ defined in \eqref{eq:les1} and the point process $\zeta_\omega^L$  defined from the uana, and for any $0 < s < 1$, we have
\beq\label{eq:les-uana-conv1}
\lim_{L \rightarrow \infty} \E \{| \xi_\omega^L [f] - \zeta_\omega^L [f]|^s \} = 0,
\eeq
for any $f \in \mathcal{A}$. Consequently, \eqref{eq:laplacetransf3}  holds.
\end{prop}

We recall the definition of the set of test functions $\mathcal{A}$ from section \ref{subsec:intensity1}, and the test function $f_\zeta \in \mathcal{A}$ given in \eqref{eq:test-fnc2}. A simple calculation shows that for $z := E_0 + \frac{\zeta}{|\Lambda_L|}$, with $\zeta = \sigma + i \tau$, with $\tau > 0$ and $\sigma \in \R$,
\beq\label{eq:green-limitAS0}
\frac{1}{|\Lambda|} {\rm Tr} \Im R_\omega^\Lambda \left( z \right) = {\rm Tr} f_\zeta(|\Lambda_L|(H_\omega^\Lambda- E_0)) = \xi_\omega^L [f_\zeta].
\eeq
Consequently, this and the linearity of \eqref{eq:special-class1} show that we must prove that for any $0 < s < 1$,
\beq\label{eq:green-limit1}
\lim_{|\Lambda| \rightarrow \infty} \E \left\{ \left| \frac{1}{|\Lambda_L|} {\rm Tr} \Im R_\omega^\Lambda  (z)
- \frac{1}{|\Lambda_L|} \sum_{p=1}^{N_L} {\rm Tr} \Im R_\omega^{\Lambda_p}  (z) \right|^s \right\} =0 ,
 \eeq
for  $z := E_0 + \frac{\zeta}{|\Lambda_L|}$.
To this end, we substitute \eqref{eq:green2} into the right side of \eqref{eq:green-limit1} the resolvent on $\Lambda_L$ and on $\Lambda_p$. This results in two types of terms.
The first identity involves the free Green's functions at
energy $z =  E_0 + \frac{\zeta}{|\Lambda_L|} \in \C \backslash [0, \infty)$:
\beq\label{eq:green-limit2}
\frac{1}{|\Lambda_L|}  \left| {\rm Tr} \Im R_0^\Lambda  (z)
-  \sum_{p=1}^{N_L} {\rm Tr} \Im R_0^{\Lambda_p}  (z) \right|.
\eeq

The second identity involves the interaction matrix $K_X^{-1}(z, \omega)$ for both the region $X = \Lambda_L$ and the regions $X = \Lambda_{\ell, p}$. As above, we let $P_m$, for $m \in \Z^d$, denote evaluation of the kernel $k(x,y)$ of an operator $K$ at the point $m$ so that $(P_m K)(x,y) = k(m,y)$ and $(K P_l)(x,y) = k(x,l)$. The second term is
\beq\label{eq:green-limit3}
 \frac{1}{|\Lambda_L|^s} \E \left\{ \left| \sum_{(l,m) \in {\tilde{\Lambda}}_L^2} \left[ {\rm Tr} [ \Im R_0^\Lambda  (z) P_k [K^{-1}_\Lambda (z, \omega)] P_m R_0^\Lambda (z) ] - \sum_{p=1}^{N_L} {\rm Tr} [ \Im R_0^{\Lambda_p}  (z) P_k [K^{-1}_{\Lambda_p} (z, \omega)] P_m R_0^{\Lambda_p} (z) ] \right] \right|^s \right\}  ,
\eeq
where we write ${\tilde{\Lambda}}_L^2$ for the Cartesian product ${\tilde{\Lambda}}_L \times {\tilde{\Lambda}}_L$.
With regard to the sum over $(k,m) \in \tilde{\Lambda_L}^2$, if either $k$ or $m$ is not in $\Lambda_{\ell, p}$ then the matrix element $[K^{-1}_{\Lambda_p} (z, \omega)]_{km}$ in the second term of \eqref{eq:green-limit3} is zero.
We now bound each term \eqref{eq:green-limit2} and \eqref{eq:green-limit3} separately.
As in section \ref{sec:les-indep-arrays1}, we work explicitly with $d = 3$, the other cases being similar and less singular, see the appendix in section \ref{sec:dimensions12}.

\subsection{Estimation of the first term \eqref{eq:green-limit2}}\label{subsec:first-term1}

To prove the vanishing as $L \rightarrow \infty$ of the term in \eqref{eq:green-limit2}, we introduce cut-off functions on subsets of $\Lambda_L$. Let ${\Lambda}^\prime_p := \{ x \in \Lambda_p ~|~ {\rm dist} ~(x, \partial \Lambda_p ) > \log L \}$. We denote by ${\Lambda^\prime_p}^c := \Lambda_p \backslash {\Lambda}^\prime_p$. In this way, we have the decomposition of $\Lambda_L$:
\bea\label{eq:decomp1}
\Lambda_L & = & \left( \bigcup_{p=1}^{N_L} \Lambda^\prime_p  \right) \cup \left( \bigcup_{p=1}^{N_L} {\Lambda^\prime_p}^c   \right) \nonumber \\
 & =: & {\Lambda}^\prime_L \cup {\Lambda^\prime_L}^c ,
 \eea
up to sets of Lebesgue measure zero.
Correspondingly, we write characteristic functions on these sets as
\beq\label{eq:decomp2}
\chi_{\Lambda_L} = \left( \sum_{p=1}^{N_L} \chi_{{\Lambda}^\prime_p} \right) + \chi_{{\Lambda^\prime_L}^c} .
\eeq
Inserting this decomposition into \eqref{eq:green-limit2}, we obtain two terms, one localized in $\Lambda^{\prime}_L$:
\beq\label{eq:green-limit4}
\frac{1}{|\Lambda_L|} \sum_{p=1}^{N_L} \left| {\rm Tr} \Im \chi_{{\Lambda}^\prime_p} \left( R_0^\Lambda  (z)
-   R_0^{\Lambda_p}  (z)  \right) \right| ,
\eeq
and one localized in the complement ${{\Lambda}^\prime_L}^c$:
\beq\label{eq:green-limit5}
\frac{1}{|\Lambda_L|} \sum_{p=1}^{N_L} \left| {\rm Tr} \Im \chi_{{{\Lambda}^\prime_p}^c} \left( R_0^\Lambda  (z)
-   R_0^{\Lambda_p}  (z)  \right) \right| ,
\eeq

We estimate \eqref{eq:green-limit4} using the fact that $(H_0^L - H_0^{\ell,p})g = 0$ if $g$ is a smooth function supported in $\Lambda_{\ell,p}$ away from $\partial \Lambda_{\ell.p}$. That is, the difference between $H_0^L$ and $H_0^{\ell,p}$ restricted to $\Lambda_{\ell,p}$ is localized near $\partial \Lambda_{\ell,p}$ whereas the trace in  \eqref{eq:green-limit4}  is localized to ${\Lambda}^\prime_{\ell, p}$ and the Green's functions decay exponentially \eqref{eq:green-exp-decay1}   since $\Re z < 0$.  We estimate \eqref{eq:green-limit5} using the fact that $|{\Lambda}_{\ell,p}^c|$ is small relative to 
$|\Lambda_{\ell,p}|$. 
The calculations are similar to the estimate of $A_L$ in part 2 of the proof of Proposition \ref{prop:theta-approx1}.
We introduce a boundary operator $\Gamma_{\ell,p}$ described by an application of Green's Theorem so 
$\Gamma_{\ell,p} R_0^{\Lambda_p}(z)$ is the restriction of $\nu \cdot \nabla R_0^{\Lambda_{\ell,p}}(z)$ to $\partial \Lambda_{\ell,p}$, where $\nu$ is the outward normal unit vector. Let $\varphi_{\ell,p}$ be a smooth function localized near $\partial \Lambda_{\ell,p}$ so that $\Gamma_{\ell,p} = \varphi_{\ell,p} \Gamma_{\ell,p}$.
Writing $Y_{\ell,p}$ for ${\Lambda}^\prime_{\ell,p}$ or $\partial {\Lambda}^\prime_{\ell,p}$, we rewrite the traces in \eqref{eq:green-limit4} and in \eqref{eq:green-limit5} using this resolvent formula in a compact notation:
\bea\label{eq:green-limit6}
| {\rm Tr} ~\{ \chi_{Y_{\ell,p}} (  R_0^{\Lambda_L} (z) - R_0^{\Lambda_{\ell,p}}(z)) \chi_{Y_{\ell,p}}  \} | & = & | {\rm Tr} ~\{ \chi_{Y_{\ell,p}}  R_0^{\Lambda_L}(z) \varphi_{\ell,p} \Gamma_{\ell,p}  R_0^{\Lambda_{\ell,p}}(z)  \chi_{{Y}_{\ell,p}} \} | \nonumber \\
 & \leq & \| \chi_{Y_{\ell,p}}  R_0^{\Lambda_L}(z) \varphi_{\ell,p} \Gamma_{\ell,p} R_0^{\Lambda_{\ell,p}}(z) \chi_{Y_{\ell,p}} \|_1 \nonumber \\
  & \leq & \| \chi_{Y_{\ell,p}}  R_0^{\Lambda_L}(z) \varphi_{{\ell,p}} \|_2
 \| \varphi_{\ell,p}  \Gamma_{\ell,p} R_0^{\Lambda_{\ell,p}}(z) \chi_{Y_{\ell,p}}\|_2 . \nonumber \\
 &  &
 \eea
For \eqref{eq:green-limit4}, we have $Y_L = {\Lambda}^\prime_{\ell,p}$ and note that
\beq\label{eq:green-limit6a}
 \| \chi_{{\Lambda}^\prime_{\ell,p}}  R_0^{\Lambda_L}(z) \varphi_{{\ell,p}} \|_2 \leq C [\ell^2 \log L]^\frac{1}{2},
 \eeq
using the exponential decay of the Green's function to control the integral over ${\Lambda}^\prime_{\ell,p}$.
A similar estimate holds for the second term since the gradient of the Green's function is integrable and decays exponentially. As a consequence, we have the bound
  \beq\label{eq:green-limit4a}
\frac{1}{|\Lambda_L|} \sum_{p=1}^{N_L}  \left| {\rm Tr} \chi_{{\Lambda}^\prime_{\ell,p}} \left( R_0^{\Lambda_L}  (z)
-   R_0^{\Lambda_{\ell,p}}  (z)  \right) \right| \leq C_0 \frac{1}{L^3} \left( \frac{L}{\ell} \right)^3
\ell^2 \log L \leq \frac{\log L}{L^\alpha} .
\eeq
Concerning \eqref{eq:green-limit5} with $Y_{\ell,p} = {\lambda}^\prime_{\ell,p}$, we have the bound
\beq\label{eq:green-limit6b}
 \| \chi_{\partial {\Lambda}^\prime_{\ell,p}}  R_0^{\Lambda_L}(z) \varphi_{{\ell,p}} \|_2 \leq C [\ell^2 \log L]^\frac{1}{2},
 \eeq
so that
\beq\label{eq:green-limit5a}
\frac{1}{|\Lambda_L|}  \left| {\rm Tr} \chi_{\partial {\Lambda}^\prime_{\ell,p}} \left( R_0^{\Lambda_L}  (z)
-   R_0^{\Lambda_{\ell,p}}  (z)  \right) \right| \leq C_0 \frac{\log L}{L^\alpha} .
\eeq
Bounds \eqref{eq:green-limit4a} and \eqref{eq:green-limit5a} show that \eqref{eq:green-limit4} and \eqref{eq:green-limit5} vanish as $L \rightarrow \infty$.

\subsection{Decomposition of the second term \eqref{eq:green-limit3}}\label{subsec:second-term-decomp1}

The operators in \eqref{eq:green-limit3} involving the interaction matrix $K^{-1}(z;\omega)$
are integral operators with singular, but square integrable, kernels. The singularities of the kernels are located at the points of $\Z^d$ in the respective cubes. Let $M_\epsilon := \cup_{j \in \tilde{\Lambda}_L } B(j; \epsilon)$, for $\epsilon > 0$ small to be chosen below. We write $M_\epsilon^c := \Lambda_L \backslash M_\epsilon$ for the complementary set. The substitution of the decomposition $1= \chi_{M_\epsilon} + (1 - \chi_{M_\epsilon})= \chi_{M_\epsilon} + \chi_{M_\epsilon^c}$ into each trace results in four terms. The two terms with the localization
$\chi_{M_\epsilon^c}$ involve the trace of operators with continuous kernels. These will be estimated in section \ref{subsec:second-term1} using the kernels of the operators. The other two terms involve singular kernels supported on $M_\epsilon$. These will be estimated here.

\begin{lemma}\label{lemma:second-term-singular1}
For the choice of $\epsilon = o(|\Lambda|^{-\frac{1}{s}})$ and $0 < s < 1$, we have
\beq\label{eq:green-second-sing1}
\lim_{L \rightarrow \infty} \frac{1}{|\Lambda_L|^s} \sum_{(k,m) \in \tilde{\Lambda}_L^2} \E \left\{ \left| {\rm Tr}  \Im \chi_{M_\epsilon} R_0^\Lambda  (z) P_k [K^{-1}_\Lambda (z, \omega)]_{km} P_m R_0^\Lambda (z) \right|^s \right\} = 0 ,
 \eeq
 and
 \beq\label{eq:green-second-sing2}
\lim_{L \rightarrow \infty} \frac{1}{|\Lambda_L|^s} \E \left\{ \left|  \sum_{p=1}^{N_L} \sum_{(k,m) \in \tilde{\Lambda}_{\ell,p}^2} {\rm Tr} \left[ \Im  \chi_{M_\epsilon} R_0^{\Lambda_{\ell,p}}  (z) P_k [K^{-1}_{\Lambda_{\ell,p}} (z, \omega)]_{km} P_m R_0^{\Lambda_{\ell,p}} (z) \right] \right|^s \right\} = 0 .
\eeq
\end{lemma}

\begin{proof}
We estimate each trace using the Hilbert-Schmidt norm and the localization estimate \eqref{eq:K-exp-decay1}. For example, for the term
\eqref{eq:green-second-sing1}, the trace is bounded above by
\bea\label{eq:green-second-sing-est1}
\lefteqn{ \frac{1}{|\Lambda_L|^s} \sum_{(k,m)\in \tilde{\Lambda}_L^2} \E \left\{  \left| {\rm Tr}  \Im \chi_{M_\epsilon} R_0^{\Lambda_L}  (z) P_k [K^{-1}_{\Lambda_L} (z, \omega)]_{km} P_m R_0^{\Lambda_L} (z) \right|^s \right\} } \nonumber \\
 & \leq  & \frac{1}{|\Lambda_L|^s} \sum_{(k,m) \in \tilde{\Lambda}_L^2} \E \left\{  \| \chi_{M_\epsilon} R_0^{\Lambda_L}  (z) P_k [K^{-1}_{\Lambda_L} (z, \omega)]_{km} P_m R_0^{\Lambda_L} (z) \chi_{M_\epsilon} \|_1^s \right\}  \nonumber \\
 & \leq & \frac{1}{|\Lambda_L|^s} \| \chi_{M_\epsilon} R_0^{\Lambda_L}  (z) P_{{\tilde{\Lambda}}_L} \|_2^{2s} ~ \sum_{(k,m) \in \tilde{\Lambda}_L^2}
 \E \{ |[K^{-1}_{\Lambda_L} (z, \omega)]_{km}|^s  \} . \nonumber \\
 & &
 \eea
A simple calculation with the integral kernel of the operator $\chi_{M_\epsilon} R_0^{\Lambda_L}  (z) P_{\tilde{\Lambda}}$ as a map from $L^2(\Lambda_L) \times \ell^2(\tilde{\Lambda}_L)$ shows that
\beq\label{eq:green-second-sing-est2}
\| \chi_{M_\epsilon} R_0^{\Lambda_L}  (z) \|_2^2 \leq \epsilon |\Lambda_L|.
\eeq
Consequently, the term on the left in \eqref{eq:green-second-sing1} is bounded above by
\beq\label{eq:green-second-sing-est12}
\epsilon^s  ~\left( \sum_{(k,m) \in \tilde{\Lambda}_L^2} \E \{ |[K^{-1}_{\Lambda_L} (z, \omega)]_{km}|^s \} \right) \leq  \epsilon^s |\Lambda|^{} ,
\eeq
and this vanishes as \rm $\epsilon = o( |\Lambda|^{-1/s})$.
As for the second trace term \eqref{eq:green-second-sing2}, we find in a similar manner
\bea\label{eq:green-second-sing-est3}
\frac{1}{|\Lambda_L|^s}  \sum_{p=1}^{N_L} \sum_{(k,m) \in \tilde{\Lambda}_{\ell,p}^2} \E \left| {\rm Tr} \left[ \Im  \chi_{M_\epsilon} R_0^{\Lambda_{\ell,p}}  (z) P_k [K^{-1}_{\Lambda_{\ell,p}} (z, \omega)]_{km} P_m R_0^{\Lambda_{\ell,p}} (z) \right] \right|^s
 & \leq &  \left[ \left( \frac{1}{L^d} \right) \left( \frac{L}{\ell} \right)^d \right]^s \epsilon^s |\Lambda_\ell|^s \nonumber \\
 & \leq & \epsilon^s .
 \eea
Taking $\epsilon = o(|\Lambda|^{-\frac{1}{s}})$ in \eqref{eq:green-second-sing-est12} and \eqref{eq:green-second-sing-est3} yields the result.

\end{proof}


\subsection{Estimation of the terms involving continuous kernels}\label{subsec:second-term1}

The decomposition result of section \ref{subsec:second-term-decomp1} allows us to compute the trace
of the two terms in \eqref{eq:green-limit3} with $\chi_{M_\epsilon^c}$ inserted into the trace by integrating the continuous kernel over the diagonal. We use localization bounds in order to estimate these integrals.

\begin{lemma}\label{lemma:second-term-nonsingular1}
For any  $0 < s < {1}$ and  $\epsilon = o(|\Lambda_L|^{\frac{1}{s}}$,
and $E_0 \in (-\infty, \tilde{E_0}) \cap \Sigma^{\rm CL}$,
 we have
\bea\label{eq:green-second-sing3}
\lefteqn{ \lim_{L \rightarrow \infty} \frac{1}{|\Lambda_L|^s} \sum_{(k,m) \in \tilde{\Lambda}_L \times \tilde{\Lambda}_L} \left\{  \E \left| {\rm Tr} [  \Im \chi_{M_\epsilon^c} R_0^\Lambda  (z) P_k [K^{-1}_\Lambda (z, \omega)]_{km} P_m R_0^\Lambda (z) ] \right. \right. } \nonumber \\
 &  - &
  \left.  \left. \sum_{p=1}^{N_L} {\rm Tr} \left[ \Im  \chi_{M_\epsilon^c} R_0^{\Lambda_p}  (z) P_k [K^{-1}_{\Lambda_p} (z, \omega)]_{km} P_m R_0^{\Lambda_p} (z) \right] \right|^s \right\} \nonumber \\
  & =  & 0 ,
\eea
where $z = E_0 + \frac{\zeta}{|\Lambda_L|}$ and $\zeta = \sigma + i \tau$ with $\tau > 0$.
\end{lemma}

\begin{proof}
An upper bound for the left side of \eqref{eq:green-second-sing3} is
\bea\label{eq:basic-eqn1}
\lefteqn{ \frac{1}{|\Lambda_L|^s} \sum_{(k,m) \in \tilde{\Lambda}_L^2}
  ~\int_{\Lambda_L \cap M_\epsilon^c} ~\E \left\{ | G_L (x,k;z) [K_L^{-1}(z,\omega)]_{km} G_L (m,x;z) \right. } \nonumber \\
   & -&   \sum_{p=1}^{N_L} G_{\ell, p} (x,k;z) [(K_{\ell, p}^{-1}(z,\omega)]_{km} G_{\ell,p} (m,x;z)|^s  \} ~d^3x.
 \eea
 Taking advantage of the decomposition $\Lambda_L = {\rm Int} \left( \overline{\cup_{p=1}^{N_L} \Lambda_{\ell, p}} \right)$, we may also  write \eqref{eq:basic-eqn1} as
 \bea\label{eq:basic-eqn2}
\lefteqn{ \frac{1}{|\Lambda_L|^s} \sum_{p=1}^{N_L} \sum_{(k,m) \in \tilde{\Lambda}_L^2}
  ~\int_{\Lambda_{\ell, p} \cap M_\epsilon^c} ~\E \left\{ | G_L (x,k;z) [K_L^{-1}(z,\omega)]_{km} G_L (m,x;z) \right. } \nonumber \\
   & -&   G_{\ell, p} (x,k;z)  [K_{\ell, p}^{-1}(z,\omega)]_{km} G_{\ell,p} (m,x;z)|^s  \} ~d^3x.
 \eea
We distinguish various cases depending on the relative position of the indices $(k,m) \in \tilde{\Lambda}_L \times \tilde{\Lambda}_L$.

\vspace{.1in}
\noindent
 {\bf Case 1:} $\| k-m \| > \sqrt{d} \ell = \sqrt{3} \ell$. In this case, the pair $(k,m)$ cannot belong to the same subcube $\Lambda_{\ell,p}$ so the matrix element $[K_{\ell,p}^{-1}(z,\omega)]_{km} = 0$. We define the set $\mathcal{A}_L := \{(k,m) ~|~ (k,m) \in \tilde{\Lambda}_L \times \tilde{\Lambda}_L, \|m-k\| > \sqrt{d} \ell \}$. As a consequence, using the localization estimate \eqref{eq:K-exp-decay1} and the explicit form of the Green's functions, the quantity \eqref{eq:basic-eqn1} becomes
\bea\label{eq:case1}
\lefteqn{ \frac{1}{|\Lambda_L|^s} \sum_{(k,m) \in {\mathcal{A}_L}} ~\int_{\Lambda_L \cap M_\epsilon^c} ~\E \left\{ | G_L (x,k;z) \left( [K_L^{-1}(z,\omega)]_{km} \right) G_{L} (m,x;z)|^s ~d^3x \right\} }  \nonumber \\
 & = & \frac{1}{|\Lambda_L|^s}  \sum_{(k,m) \in {\mathcal{A}_L}} ~\E \left\{ |[K_L^{-1}(z,\omega)_{km}|^s \right\} ~\int_{\Lambda_L \cap M_\epsilon^c} ~ | G_L (x,k;z) G_{L} (m,x;z)|^s ~d^3x \nonumber \\
  & \leq &   \frac{1}{|\Lambda_L|^s} \sum_{(k,m) \in {\mathcal{A}_L}} e^{- s \gamma_{s,3} (z) \| k-m\|} \epsilon^{-2s} ~\int_{\Lambda_L \cap M_\epsilon^c} ~ e^{- c(E_0)(\|x - k\| + \|x-m\|)} ~d^3x \nonumber \\
 & \leq &   |\Lambda_L|^{s+1}  e^{- \gamma_{s,3} (z)  \sqrt{3} \ell}.
  \eea
Since $\ell = L^\alpha$, for $0 < \alpha < 1$, this term vanishes as $L \rightarrow \infty$.

\vspace{.1in}
 \noindent
 {\bf Case 2:} For $C > 0$ sufficiently large, $C \log L < \| k-m \| \leq \sqrt{d} \ell = \sqrt{3} \ell$. We must distinguish two cases: (2a) The pair $(k,m)$ belong to different subcubes; (2b) The pair $(k,m)$ belong to the same subcube, say $\Lambda_{\ell,q}$. The case (2a) is the same as case (1) since in this case $[K_{\ell,p}^{-1}(z,\omega)]_{km} = 0$, for any $p=1, \ldots, N_L$. As for the case (2b), we define the set $\mathcal{B}_q := \{(k,m) ~|~ (k,m) \in \Lambda_{\ell,q}, \|m-k\| > C \log L \}$. Hence, the contribution to \eqref{eq:basic-eqn2} from all the regions $\mathcal{B}_q$ is
\bea\label{eq:basic-eqn4}
\lefteqn{ \frac{1}{|\Lambda_L|^s} \sum_{q = 1}^{N_L} ~\sum_{(k,m) \in \mathcal{B}_q}
  ~\int_{\Lambda_{\ell,q} \cap M_\epsilon^c} ~\E \left\{ | G_L (x,k;z) [K_L^{-1}(z,\omega)]_{km} G_L (m,x;z) \right. } \nonumber \\
  & - &     \left. G_{\ell, q} (x,k;z)  (K_{\ell, q}^{-1})_{km} G_{\ell,q} (m,x;z)|^s  \right\} ~d^3x \nonumber \\
   &=&  \frac{1}{|\Lambda_L|^s} \sum_{q = 1}^{N_L} ~\sum_{(k,m) \in \mathcal{B}_q}
  ~\int_{\Lambda_{\ell,q} \cap M_\epsilon^c} ~\E \left\{ | G_L (x,k;z)[ K_L^{-1}(z,\omega) - K_{\ell,q}^{-1} (z,\omega)]_{km} G_{\ell,q} (m,x;z)|^s d^3 x \right\} \nonumber \\
    & & + \frac{1}{|\Lambda_L|^s} \sum_{q = 1}^{N_L} ~\sum_{(k,m) \in \mathcal{B}_q} [ \mathcal{E}_{q,1} (k,m; L) + \mathcal{E}_{q,2}(k,m;L) + \mathcal{E}_{q,3}(k,m;L)] ,
\eea
where we write
\bea\label{eq:errors1}
\mathcal{E}_{q,1}(k,m;L) & := &   \int_{\Lambda_{\ell,q} \cap M_\epsilon^c} ~\E \left\{ | G_L (x,k;z)
[ K_L^{-1}(z,\omega) ]_{km}[ G_L (m,x;z) - G_{\ell,q} (m,x;z)]|^s d^3 x \right\} \nonumber \\
\mathcal{E}_{q,2}(k,m;L)  & := &  \int_{\Lambda_{\ell,q} \cap M_\epsilon^c} ~\E \left\{ | [G_{\ell,q}(x,k;z) - G_L (x,k;z)][  K_{\ell,q}^{-1}(z,\omega) ]_{km} G_{\ell,q} (m,x;z)|^s d^3 x \right\}  \nonumber \\
\mathcal{E}_{q,3}(k,m; L) &   := &   \int_{\Lambda_L \backslash \Lambda_{\ell,q} \cap M_\epsilon^c} ~\E \left\{ | G_L (x,k;z)[ K_L^{-1} (z,\omega)]_{km} G_{L} (m,x;z)|^s d^3 x \right\}
\eea
We estimate each of these terms as follows.
In order to estimate the main term in \eqref{eq:basic-eqn4}, we apply the second resolvent equation to the operators $K_L$ and $K_{\ell,q}$ in the following form
\bea\label{eq:sec-resolv-K}
\E \{ | [K_L^{-1}(z,\omega) - K_{\ell, q}^{-1}(z,\omega)]_{km}|^s \} & \leq & \sum_{(n,r) \in \tilde{\Lambda}_\ell^2}
\E \{ | [K_L^{-1}(z,\omega)]_{kn} [ K_L(z,\omega) - K_{\ell,q}(z,\omega) ]_{nr} [K_{\ell,q}^{-1}(z,\omega)]_{rm}|^s \} \nonumber \\
  & \leq &  \sum_{(n,r) \in \tilde{\Lambda}_\ell^2} |[ t_L(z) - t_{\ell,q}(z) ]_{nr}|^s \E \{
  | K_L^{-1}(z,\omega)]_{kn}|^s |[K_{\ell,q}^{-1}(z,\omega)]_{rm}|^s \}  \nonumber \\
  & &
\eea
where deterministic matrix $t_L(z)$ is defined in \eqref{eq:matrixKE1}. We use the bound for a nonnegative random variable $X$
$$
\E \{ X^s \} = \E \{ X^{\frac{s}{2}} X^{\frac{s}{2}} \} \leq \|X^{\frac{s}{2}} \|_\infty  \E \{ X^{\frac{s}{2}} \},
$$
with $X = | (K_L^{-1})_{kn}|~ |(K_{\ell,q}^{-1})_{rm}|$ so that $\|X^{\frac{s}{2}} \|_\infty \sim |\Lambda_L|^s$.
Because of the structure of the matrix $M_{\ell,q}(z) := t_L(z) - t_{\ell,q}(z)$ and the fact that $z < 0$, the matrix elements decay exponentially away from $\partial \Lambda_{\ell,q}$.
Inserting the bound \eqref{eq:sec-resolv-K} into the integral in \eqref{eq:basic-eqn4}, we obtain the bound
 \beq\label{eq:case2b}
\sum_{q=1}^{N_L} \sum_{(k,m) \in \mathcal{B}_q}
    \sum_{(n,r) \in \tilde{\Lambda}_{\ell,q}^2} |[ M_{\ell,q}(z)]_{nr}|^s \E \{ |[K_L^{-1}(z,\omega)]_{kn}|^{s}\}^{\frac{1}{2}} \E \{ |[K_{\ell,q}^{-1}(z,\omega)]_{rm}|^{s} \}^{\frac{1}{2}} I_{km}(L,\ell,q).
\eeq
where we write
\beq\label{eq:int-green1}
I_{km}(L,\ell,q) := \int_{\Lambda_{\ell,q} \cap M_\epsilon^c} ~  |G_L (x,k;z) G_{\ell,q} (m,x;z)|^s  ~d^3x .
\eeq
The localization estimate \eqref{eq:K-exp-decay1} provides the bound
\beq\label{eq:K-exp-decay2}
\E \{ |[K_L^{-1}(z,\omega)]_{kn}|^{s}\}^{\frac{1}{2}} \E \{ |[K_{\ell,q}^{-1}(z,\omega)]_{rm}|^{s} \}^{\frac{1}{2}}
 \leq C_s^2 e^{-s \gamma_{s,3}(z) ( \|k-n \| + \|r-m\| )}.
 \eeq
This yields the bound
\beq\label{eq:sum-bound1}
\sum_{(n,r)\in \tilde{\Lambda}_{\ell,q}^2} C_s^2 e^{-s \gamma_{s,3}(z) ( \|k-n \| + \|r-m\| )} \leq C_1 .
\eeq
By standard estimates on the local Green's functions \eqref{eq:green-exp-decay1}, we have
\beq\label{eq:greens-fnc-decay1}
\int_{\Lambda_{\ell,q} \cap M_\epsilon^c} ~  |G_L (x,k;z) G_{\ell,q} (m,x;z)|^s  ~d^dx
 \leq \int_{\Lambda_{\ell,q} \cap M_\epsilon^c} ~  \frac{e^{- c(E_0) \| k-x\|}}{\| k-x\|^s}  \frac{e^{- c (E_0) s \| m-x\|}}{\| m-x\|^s} ~d^3x.
\eeq
Because the kernel is locally integrable we find that
\beq\label{eq:greens-fnc-decay2}
\int_{\Lambda_{\ell,q} \cap M_\epsilon^c} ~  |G_L (x,k;z) G_{\ell,q} (m,x;z)|^s  ~d^3x
\leq
e^{- \tilde{c}(E_0) s \| k-m\|} ,
\eeq
so that
\beq\label{eq:sum-bound2}
\sum_{(k,m) \in \mathcal{B}_q} e^{- \tilde{c}(E_0) s \| k-m\|} \leq | \mathcal{B}_q| e^{- \tilde{c}(E_0) s C \log L}.
\eeq
since $| \mathcal{B}_q| \sim \mathcal{O}(\ell^d)$, these estimates show that for $C > 0$ large, the main term in \eqref{eq:basic-eqn4} vanishes as $L \rightarrow \infty$.

\vspace{.1in}
\noindent
 {\bf Case 3:} $0 < \| k-m\| < C \log L$. To deal with these points, we define the security zone $\mathcal{Z}_\ell = \{ m \in \Z^d \cap \Lambda_L ~|~ {\mbox dist}(x \cup_p \partial \Lambda_{\ell,p}) < 2C \log L$, for $C > 0$ sufficiently large. The volume of the security zone $\mathcal{Z}_\ell$ is
\beq\label{eq:zone1}
| \mathcal{Z}_\ell| = \left( \frac{L}{\ell} \right)^d \ell^{d-1} 2C \log L = 2C L^{d(1 - \alpha)} \log L.
  \eeq
We distinguish three cases. Case (3a) consists of pairs $(k,m) \in \Lambda_{\ell,q}^2$ so that at least one element of the pair is a distance at least $2C \log L$ from $\partial \Lambda_{q, \ell}$ and the other point belongs to
$\Lambda_{\ell,q} \backslash ( \Lambda_{\ell,q} \cap \mathcal{Z}_\ell$.
Case (3b) consists of those pairs $\mathcal{Z}_{\ell,q}$ belonging to $\mathcal{Z}_\ell \cap \Lambda_{\ell,q}$ for one $q$,
   and case (3c) consists of those pairs in the security zone belonging to two separate cubes $\Lambda_{\ell,p}$ and $\Lambda_{\ell,q}$, for $p \neq q$, and denoted by $\mathcal{Z}_{\ell,p,q}$.

\vspace{.1in}
\noindent
{\bf Case 3a.} $(m,k) \in \Lambda_{\ell, q} \backslash (\Lambda_{\ell, q} \cap \mathcal{Z}_\ell)$ and $0 < \| k-m \| \leq C \log L$. We denote the set of such pairs $\mathcal{C}_q$. Since the points $(k,m)$ are restricted to one cube, the estimates are similar to those for case (2b).
As in part 5 of the proof of Proposition \ref{prop:theta-approx1}, we introduce two sets of indices $(r,n)$: $\tilde{X}_1 := \{ (n,r) \in \tilde{\Lambda}_{\ell,q}^2 \}$ and $\tilde{X}_2 := \{ (n,r) \in (\tilde{\Lambda}_L \backslash \tilde{\Lambda}_{\ell,q}) \} \times \tilde{\Lambda}_{\ell,q}$.  The main term to estimate is
  \beq\label{eq:case3a-1}
\frac{1}{|\Lambda_L|^{\frac{s}{2}}} \sum_{q=1}^{N_L} \sum_{(k,m) \in \mathcal{C}_q}
    \sum_{(n,r) \in \tilde{X}_1 \cup \tilde{X}_2} |[ M_{\ell,q}(z)]_{nr}|^s \E \{ [K_L^{-1}(z,\omega)]_{kn}|^{s}\}^{\frac{1}{2}} \E \{ |[K_{\ell,q}^{-1}(z,\omega)]_{rm}|^{s} \}^{\frac{1}{2}} I_{km}(L,\ell,q).
\eeq
where we write
\beq\label{eq:int-green21}
I_{km}(L,\ell,q) := \int_{\Lambda_{\ell,q} \cap M_\epsilon^c} ~  |G_L (x,k;z) G_{\ell,q} (m,x;z)|^s  ~d^3x .
\eeq
To control the sum over $(n,r)$, we note that the difference $M_{\ell,q}(z) := t_L(z) - t_{\ell,q}(z)$
is localized on $\partial \Lambda_{\ell,q}$ and decays exponentially fast away from the boundary since $E_0 < 0$.
We divide the boundary zone of $\partial \Lambda_{\ell,q}$ into two regions: $\partial \Lambda_{\ell,q}^{(1)}$
consisting of all lattice points for which ${\rm dist} (s, \partial \Lambda_{\ell,q}) < C \log L$ and $\partial \Lambda_{\ell,q}^{(2)}$ consisting of all lattice points $t$ so that ${\rm dist} (t, \partial \Lambda_{\ell,q}) > C \log L$. We then have that for $(n,r)$, such that either point is in zone 1, the matrix elements $| [M_{\ell,q}(z)]_{nr}|$ are uniformly bounded and $\| n-k\|$ or $\|r-m\|$ is uniformly bounded below by $C \log L$.
This results in exponential decay from localization:
\beq\label{eq:K-exp-decay22}
\E \{ |[K_L^{-1}(z,\omega)]_{kn}|^{s}\}^{\frac{1}{2}} \E \{ |[K_{\ell,q}^{-1}(z,\omega)]_{rm}|^{s} \}^{\frac{1}{2}}
 \leq C_s^2 e^{-s \gamma_{s,3} (z) C \log L}.
 \eeq
For zone 2, we have the bound
\beq\label{eq:bdry-green-decay-3a}
\| [ M_{\ell,q}(z)]_{nr} \| \leq | \partial \Lambda_{\ell,q}^{(1)} |^2 e^{- c(E_0) C \log L},
\eeq
coming from the exponential decay of $ M_{\ell,q}(z)$ away from $\partial \Lambda_{\ell,q}$.
By standard estimates on the local Green's functions we have
\beq\label{eq:greens-fnc-decay112}
I_{km}(L,\ell,q)  \leq \int_{\Lambda_{\ell,q} \cap M_\epsilon^c} ~  \frac{e^{- c (E_0) s \| k-x\|}}{\| k-x\|^s}  \frac{e^{- c (E_0) s \| m-x\|}}{\| m-x\|^s} ~d^3x \leq e^{- \tilde{c}(E_0) s \| k-m\|} .
\eeq
Combining \eqref{eq:K-exp-decay22}, \eqref{eq:bdry-green-decay-3a}, and \eqref{eq:greens-fnc-decay112}, we find that the $\tilde{X}_1$ contribution to \eqref{eq:case3a-1-1} is bounded above by
\beq\label{eq:case3a-1-1}
L^{-\frac{3s}{2}}(L^{3(1 - \alpha)})|\partial \Lambda_{\ell,q}|^2 ~e^{- c(E_0) C \log L}.
\eeq
that vanishes as $L \rightarrow \infty$. As for the $\tilde{X}_2$ contribution to \eqref{eq:case3a-1-1}, we note that $\| k - n \| > C \log L$. Consequently, the sum over $n \in (\tilde{\Lambda}_L \backslash \tilde{\Lambda}_{\ell,q}) \}$ provides an upper bound of the form $e^{- c(E_0) C \log L}$ and the entire contribution vanishes as $L \rightarrow \infty$.

\vspace{.1in}
\noindent
 {\bf Case 3b:} $(m,k) \in \mathcal{Z}_{\ell}$, so that for some $q = 1, \ldots, N_L$, we have $k,m \in  \Lambda_{\ell,q}$, for some $q$, with ${\rm dist}(m, \partial \Lambda_{\ell,q}) < 2 C \log L$ and ${\rm dist} (k, \partial \Lambda_{\ell,q}) < 2 C \log L$, and $0 < \| k-m \| \leq C \log L$. We denote such pairs in $\Lambda_{\ell,q}$ by $\mathcal{Z}_{{\ell,q}}$.
Since $\mathcal{Z}$ is close to $\partial \Lambda_{\ell,q}$, we need to estimate two integrals
\beq\label{eq:case3b-1}
\frac{1}{| \Lambda_L|^{{s}}} \sum_{q=1}^{N_L} \sum_{(k,m) \in \mathcal{Z}_{\ell,q; s}}
   \E \{ |[K_L^{-1}(z,\omega)]_{km}|^{s} \}
 \left| \int_{\Lambda_{\ell,q}} ~ G_L (x,k;z) G_{\ell,q} (m,x;z) ~d^dx \right|^s ~d^3 x ,
\eeq
and
\beq\label{eq:case3b-2}
\frac{1}{| \Lambda_L|^{{s}}} \sum_{q=1}^{N_L} \sum_{(k,m) \in \mathcal{Z}_{\ell,q; s}}
   \E \{ |[K_{\ell,q}^{-1}(z,\omega)]_{km}|^{s} \}
 \left| \int_{\Lambda_{\ell,q}} ~ G_L (x,k;z) G_{\ell,q} (m,x;z) ~d^dx \right|^s ~d^3 x .
\eeq
Due to the exponential decay of the Green's functions at negative energies, the integrals in \eqref{eq:case3b-1} and \eqref{eq:case3b-2}
are uniformly bounded in $L$.
We use the localization estimate to bound the expectations in \eqref{eq:case3b-1} and \eqref{eq:case3b-2} so the sum over these expectations is bounded by $|\mathcal{Z}_{\ell,q}|$. As a result, we obtain a bound of the following form for both \eqref{eq:case3b-1} and \eqref{eq:case3b-2}:
\beq\label{eq:case3b-3}
L^{- {ds}} L^{d(1-\alpha)} \sum_{(k,m) \in \mathcal{Z}_{\ell,q; s}}  e^{- s\alpha_z \|k-m\|}
\leq L^{- {ds}} L^{d(\alpha - 1)} L^{\alpha (d-1)} \log L .
\eeq
We require $d(1-s) - \alpha$ be negative. For $d=3$, provided $3(1-s) - \alpha < 0$ and requiring $\frac{2}{3} < s < 1$, there exists $0 < \alpha < 1$ so this condition is obtained. Consequently, for these choices, these terms vanish as $L \rightarrow \infty$.

\vspace{.1in}
\noindent
 {\bf Case 3c:} $(m,k) \in \mathcal{Z}_{\ell}$, so $0 < \| k-m \| \leq 2 C \log L$ and $(k,m) \in  \mathcal{Z}_{\ell,p,q} \subset \Lambda_{\ell,q} \times \Lambda_{\ell,p}$, for some $q \neq p$.
Consequently, the matrix elements $[K^{-1}_{\ell,p}(z,\omega)]_{km} = 0$ and only the term containing $\Lambda_L$ in \eqref{eq:basic-eqn1} remains:
\bea\label{eq:case3c}
\lefteqn{ \frac{1}{|\Lambda_L|^{{s}}} \sum_{ (k,m) \in \mathcal{Z}_{\ell,d}} \E \{ | {\rm Tr} \Im [ R_0^\Lambda (z)P_k K_L^{-1} (z, \omega) P_m R_0^\Lambda (z)] |^s \} } \nonumber \\
 &\leq& \frac{1}{L^{ds}} \sum_{ (k,m) \in \mathcal{Z}_{\ell,d}} \E \left\{ |[(K_L^{-1}(z,\omega)]_{km} |^s \right\} \left(  ~\int_{\Lambda_L} ~ |G_0^L (x,k;z) G_0^L (m,x;z)|^s ~d^dx \right).
\eea
The integral over $\Lambda_L$ is uniformly bounded in $L$ due to the exponential decay of the Green's functions. 
\beq\label{eq:zone-different1}
L^{-ds} | \mathcal{Z}| = C L^{-sd} L^{(1-\alpha)d} L^{\alpha (d-1)} \log L,
\eeq
so making the same choices as in case (3b), this term vanishes as $L \rightarrow \infty$.

We finally estimate the three error terms in \eqref{eq:errors1}. The estimates for $\mathcal{E}_{q,1}$ and $\mathcal{E}_{q,2}$ are similar. After taking the expectation of the $s^{th}$-power each term, we use the localization estimate for the $\E \{ |[K^{-1}(z,\omega)]_{km}|^s \}$-terms and the fact that the Green's function decay exponentially away from the boundary. As for $\mathcal{E}_{q,3}$, we again take the $s^{th}$-power and use the localization estimate to control
$\E \{ |[K_L^{-1}(z,\omega)]_{km}|^s \}$. The Green's functions decay exponentially away from $\Lambda_{\ell, q}$ where $m$ and $k$ are located. The volume of a security zone surrounding $\partial \Lambda_{\ell, q}$ is used to bound the integral over this small region.
\end{proof}

\section{Conclusion of the proof of Theorem \ref{thm:main1} on local eigenvalue statistics}\label{sec:sp-stat1}
\setcounter{equation}{0}

  To summarize, we know that the local point processes $\xi^{\Lambda_L}_\omega$ and
  $\zeta_\omega^{\Lambda_L}$ have the same limit point. The limit of the $uana$ associated with $\zeta_\omega^{\Lambda_L}$
  is analysed using the Wegner and Minami estimates. In section \ref{sec:les-indep-arrays1}, is was shown that the Wegner and Minami estimates imply the conditions necessary to insure that the limiting point process is a Poisson point process with intensity measure $n(E_0) ds$, for $E_0 \in (-\infty,\tilde{E_0}) \cap \Sigma^{\rm CL}$ and provided $n(E_0) > 0$. The existence of the density of states for $H_\omega$
was shown in \cite[Corollary 7.3]{hkk1}.

%
%

%


\section{Appendix: Estimates on rank one perturbations}\label{sec:trace-est1}
\setcounter{equation}{0}

We present here a possibly new approach to rank one perturbations. We prove in generality that the eigenvalue counting function ${\rm Tr} E_A(I)$ of a self-adjoint operator A with locally discrete spectrum in an interval $I$  changes by at most one under by a self-adjoint rank one perturbation $B$.

\begin{prop}\label{prop:rank-one1}
Let $A$ and $B$ be self-adjoint operators with $B$ rank one. Let $I := [a,b] \subset \R$ be an interval with $\sigma(A) \cap [a,b]$ consisting of at most finitely many real eigenvalues. Then,
\begin{enumerate}
\item The traces of the spectral projectors differ by at most one:
\beq\label{eq:rank-one-pert1}
| {\rm Tr} [ E_A(I) ] - {\rm Tr} [ E_{A+B}(I) ] | \leq 1 .
\eeq

\item If ${\rm Tr} [ E_A(I) ] \geq 1$, we have

\beq\label{eq:rank-one-pert2}
0 \leq {\rm Tr} [ E_A(I) ] - 1  \leq  {\rm Tr } [ E_{A+B}(I) ]  .
\eeq

\end{enumerate}
\end{prop}


By a unitary transformation, if necessary, we may assume that $B = \Pi_{\varphi}$, the projection onto a normalized vector $\varphi$, with $\| \varphi \| =1$.
By a standard reduction, it suffices to consider the case for which $\varphi$ is $A$-cyclic.
For let $\mathcal{H}_\varphi$ be the cyclic subspace for $A$ and $\varphi$. We then have $\mathcal{H} = \mathcal{H}_\varphi \oplus \mathcal{H}_\varphi^\perp$. The subspace $\mathcal{H}_\varphi$ is also $B$-invariant and $B$ annihilates $\mathcal{H}_\varphi^\perp$ so $\sigma (A | \mathcal{H}_\varphi^\perp ) = \sigma ( (A +B) |\mathcal{H}_\varphi^\perp)$. consequently, we have
 $$
 {\rm Tr} [ E_A(I) ] - {\rm Tr} [ E_{A+B}(I) ] = {\rm Tr}_{\mathcal{H}_\varphi^A} [ E_{A|\mathcal{H}_\varphi^A}(I) ] - {\rm Tr} [ E_{(A+B)|{\mathcal{H}_\varphi^A}}(I) ].
 $$
So we now assume that $\varphi$ is cyclic for $A$ and therefore for $A+B$ (see, for example, \cite[Lemma 3.1.2]{demuthkrishna1}).
We define two measures
\beq\label{eq:defn-measures1}
\mu_A^\varphi (\cdot) := \langle \varphi, E_A(\cdot) \varphi \rangle, ~~~ {\mbox and} ~~~
\mu_{A+B}^\varphi (\cdot) := \langle \varphi, E_{A+B}(\cdot) \varphi \rangle.
\eeq

%
%
%

\begin{lemma}\label{lem:measure-complete1}
For any $x \in \sigma (A) \cap [a,b]$, we have
\beq\label{eq:complete1}
\mu_A^\varphi (\{ x\}) \neq 0 ,
\eeq
and similarly, for any $y \in \sigma (A+B) \cap [a,b]$, we have
\beq\label{eq:complete2}
\mu_{A+B}^\varphi (\{ y\}) \neq 0 .
\eeq
\end{lemma}

\begin{proof}
For \eqref{eq:complete1}, suppose that $\mu_A^\varphi (\{ x\}) = 0$.
If $\{ \psi_j \}$ are the orthonormal eigenfunctions satisfying $A \psi_j = x \psi_j$,
we have
$$
\mu_A^\varphi (\{ x\}) = \sum_j | \langle \psi_j, \varphi \rangle |^2 =0 ,
$$
so $\langle \psi_j, \varphi \rangle = 0$ for all $j$. But the $A$-cyclicity of $\varphi$ means that $\psi_j = 0$ for all $j$, a contradiction.
A similar proof applies to $A+B$ since $\varphi$ is also $A+B$-cyclic.
\end{proof}

Now suppose $\{ x_1, \ldots, x_k \}$ are the eigenvalues of $A$ in $[a,b]$ and $\{y_1, \ldots, y_\ell \}$ are the eigenvalues of $A+B$ in $[a,b]$. From Lemma \ref{lem:measure-complete1}, we have
\bea\label{eq:complete3}
\mu_A^\varphi (\{ x_i \}) \neq 0 &  \forall & i=1, \ldots, k.  \nonumber \\
\mu_{A+B}^\varphi (\{ y_j \}) \neq 0 &  \forall & j=1, \ldots, \ell.
\eea
It follows that the functions
\beq\label{eq:unperturb-map1}
F_A(z) := \langle \varphi, R_A(z) \varphi \rangle ,
\eeq
and
\beq\label{eq:unperturb-map2}
F_{A+B}(z) := \langle \varphi, R_{A+B}(z) \varphi \rangle  .
\eeq
have poles in $[a,b]$ precisely at the eigenvalues $\{ x_i \}$ and at $\{ y_j \}$.

\begin{lemma}\label{lam:bijective-maps1}
The map
\beq\label{eq:unperturb-map-bij1}
F_A(E) := \langle \varphi, R_A(E) \varphi \rangle
\eeq
restricted to each interval $F_A : (x_i, x_{i+1}) \rightarrow \R$ is bijective for all $i=1, \ldots, k-1$.
Similarly, the map
\beq\label{eq:unperturb-map-bij2}
F_{A+B}(E) := \langle \varphi, R_{A+B}(E) \varphi \rangle
\eeq
restricted to each interval $F_{A+B} : (y_j, y_{j+1}) \rightarrow \R$ is bijective for all $j=1, \ldots, \ell-1$.
\end{lemma}

\begin{proof}
The maps $F_X$, for $X=A$ or $X=A+B$, are real-valued and differentiable on each interval and the derivative satisfies
$$
F_X^\prime (z) = \langle \varphi, R_X(z)^2 \varphi \rangle > 0 .
$$
Hence, the function $F_A$, respectively, $F^{A+B}$, is strictly monotone increasing on $(x_i, x_{i+1})$, respectively,
on $(y_i, y_{i+1})$. Furthermore, since $x_i$ and $x_{i+1}$, respectively, $y_i$ and $y_{i+1}$
are poles of $F_A$, respectively, of $F^{A+B}$, we have
\beq\label{eq:ev-limits1}
\lim_{x \rightarrow x_i^+} F_A(x) = - \infty, ~~~{\mbox and } ~~~ \lim_{x \rightarrow x_{i+1}^-} F_A(x) =  \infty,
\eeq
and similarly for $F^{A+B}$ and $y_i$ and $y_{i+1}$.
This establishes that $F_X$ is a bijection from each interval to $\R$.
\end{proof}

\begin{lemma}\label{lem:intertwine1}
The poles of $F_A$ and $F_{A+B}$ in $[a,b]$ are intertwined. In each interval $(x_i, x_{i+1})$, there is exactly one pole $y_j$ of $F_{A+B}$ and in each interval $(y_j, y_{j+1})$, there is exactly one $x_i$.
\end{lemma}

\begin{proof}
\noindent
1. The lemma follows from the strict monotonicity of each function $F_A$ and $F_{A+B}$ in the intervals $(x_i, x_{i+1})$
and $(y_j, y_{j+1})$, respectively, together with the rank one perturbation formula:
\beq\label{eq:rank-one-formula1}
F_{A+B}(x) = \frac{1}{1+ \frac{1}{F_A(x)}} .
\eeq
Suppose $E \in (x_i, x_{i+1})$. From \eqref{eq:ev-limits1}, we find
\beq\label{eq:ev-limits2}
\lim_{E \rightarrow x_i^+} F_{A+B}(E) = 1, ~~~{\mbox and } ~~~ \lim_{E \rightarrow x_{i+1}^-} F_{A+B}(E) = 1.
\eeq
By the monotonicity of $F_{A+B}$, there must be an eigenvalue of $A+B$ in $(x_i, x_{i+1})$, say $y_1$. Since
$F_A(E)$ satisfies $-1 < F_A(E) < + \infty$ for $E \in (y_1, x_2)$, the perturbation formula \eqref{eq:rank-one-formula1} shows that there cannot be another eigenvalue of $A+B$ in this interval. a similar argument applies to the interval $(x_1, y_1)$. hence, each interval $(x_i, x_{i+1})$ contains precisely one eigenvlaue of $A+B$.

\noindent
2. Let us suppose that there are $k$ eigenvalues of $A$ in $[a,b]$:
\beq\label{eq:Aev1}
a < x_1 < x_2 < \cdots < x_k < b .
\eeq
From part 1, there are $k-1$ eigenvalues $y_j \in (x_i, x_{i+1})$ of $A+B$, for $i=1, \ldots, k$.
There may be an eigenvalue of $A+B$ in $[a, x_1)$ or $(x_k, b]$. In this case, the number of eigenvalues $\ell$ of $A+B$ in $[a,b]$ is $\ell \leq k-1 + 2 = k+1$, so $|\ell - k| \leq 1$. Or, there may be no eigenvalues of $A+B$ in
these two intervals in which case we also have $|\ell - k| \leq 1$.

\end{proof}


\section{Appendix: Details for dimensions $d=1$ and $d=2$}\label{sec:dimensions12}
\setcounter{equation}{0}

We provided the formulae and other details for dimensions $d=1$ and $d=2$. This information is taken from \cite{hkk1}. 
The effective energy $e_d(z)$ for $d=1,2$ is given by
\bea\label{eq:effective-energy2}
e_1(z) &=& -\frac{i}{2 \sqrt{z}} \nonumber \\
e_2(z) &=& \frac{\log \sqrt{-z}}{2 \pi} \nonumber \\
\eea
The square root function is defined with the branch cut along the positive real axis.
The effective coupling constants $\alpha_{d,k}$ are defined as $\alpha_{1,k} = - \alpha_k$ and
equal to $\alpha_k$ for $d=2$.

The free Green's functions $G_0$ on $\R^d$, for $d=1,2$, corresponding for use in \eqref{eq:green1},
are given as follows. For $d=1$, we have
\beq\label{eq:greend1}
G_0 (x-y;z) = \frac{i}{2 \sqrt{z}} e^{i \sqrt{z} | x-y |}, ~~~x,y \in \R   ,
\eeq
and for $d = 2$:
\beq\label{eq:greend2}
G_0 (x-y;z) = \frac{i}{4} H_0^{(1)}( \sqrt{z} \| x-y \|),  ~~~x,y \in \R^2 .
\eeq


%

\end{document}